\documentclass[a4paper,11pt]{amsart}
\usepackage[a4paper,margin=25mm]{geometry}

\usepackage{amsmath,amsthm}
\usepackage[T1]{fontenc}
\usepackage{enumerate}
\usepackage{hyperref}
\usepackage{lmodern}
\usepackage[cp1250]{inputenc}

\usepackage{color}
\definecolor{myYellow}{rgb}{0.9,0.9,0}

\newcommand{\flag}{\textnormal{\textsl{flag}}}
\newcommand{\mynew}{\textnormal{\textsl{new}}}
\newcommand{\next}{\textnormal{\textsl{next}}}
\newcommand{\algofont}[1]{\textnormal{\selectfont\sffamily#1}}
\newcommand{\algmain}{\algofont{TtoG}}
\newcommand{\algmaini}{\algofont{TtoGImp}}
\newcommand{\alggreedypairs}{\algofont{GreedyPairs}}
\newcommand{\algpairncr}{\algofont{PairCompNCr}}
\newcommand{\algpaircr}{\algofont{PairComp}}
\newcommand{\algpop}{\algofont{Pop}}
\newcommand{\algblocks}{\algofont{BlockCompNCr}}
\newcommand{\algblocksc}{\algofont{BlockComp}}
\newcommand{\algremblocks}{\algofont{RemCrBlocks}}
\newcommand{\RePair}{\algofont{RePair}}
\newcommand{\Mref}[1]{(\hyperref[M#1]{M#1})}
\newcommand{\Mrefall}{{\Mref{1}--\Mref{3}}}

\newcommand{\CPref}[1]{(\hyperref[CP#1]{CP#1})}

\newcommand{\PC}{{PC}}
\newcommand{\twodots}{\mathinner{\ldotp\ldotp}}
\DeclareMathOperator{\eval}{val}
\DeclareMathOperator{\weight}{w}

\newcommand{\algradix}{\algofont{RadixSort}}

\newcommand{\mycount}{\textnormal{\textsl{count}}}
\newcommand{\leftl}[1][(a)]{\ensuremath{\textnormal{\textsl{left}}#1}}
\newcommand{\mysize}{\ensuremath{\textnormal{\textsl{size}}}}

\newcommand{\rightl}[1][(a)]{\ensuremath{\textnormal{\textsl{right}}#1}}
\newcommand{\makeset}[2]{\ensuremath{ \{ #1 \: | \: #2 \} }}

\newcommand{\mytext}{\ensuremath{T}}
\newcommand{\textci}{\ensuremath{T'}}

\providecommand{\Ocomp}{\mathcal{O}}
\providecommand{\size}{\ensuremath{\Ocomp(|\mytext|)}}

\newtheorem{theorem}{Theorem}
\newtheorem{lemma}{Lemma}
\newtheorem{corollary}{Corollary}
\theoremstyle{definition}

\theoremstyle{remark}

\newtheorem{clm}{Claim}

\usepackage[ruled]{algorithm}
\usepackage[noend]{algpseudocode}

\begin{document}

\title{Approximation of grammar-based compression via recompression}

\author[A.\ Je\.z]{Artur Je\.z}
\thanks{Supported by NCN grant number 2011/01/D/ST6/07164, 2011--2014.}
\address{
Max Planck Institute f\"ur Informatik,\\
Campus E1 4,  DE-66123 Saarbr\"ucken, Germany\\
\and
Institute of Computer Science, University of Wroc{\l}aw \\
ul.\ Joliot-Curie~15, 50-383 Wroc{\l}aw, Poland\\
\texttt{aje@cs.uni.wroc.pl}}

\begin{abstract}
In this paper we present a simple linear-time algorithm
constructing a context-free grammar of size $\Ocomp(g \log (N/g))$
for the input string, where $N$ is the size of the input string
and $g$ the size of the optimal grammar generating this string.
The algorithm works for arbitrary size alphabets,
but the running time is linear assuming that the alphabet $\Sigma$ of the input string
can be identified with numbers from $\{1,\ldots , N^c \}$ for some constant $c$.
Otherwise, additional cost of $\Ocomp(n \log |\Sigma|)$ is needed.

Algorithms with such an approximation guarantee and running time are known,
the novelty of this paper is a particular simplicity of the algorithm
as well as the analysis of the algorithm, which uses a general
technique of recompression recently introduced by the author.
Furthermore, contrary to the previous results, this work does not use the LZ representation
of the input string in the construction, nor in the analysis.
\end{abstract}
\keywords{Grammar-based compression; Construction of the smallest grammar; SLP; compression}
	\maketitle

\section{Introduction}
	\label{sec:intro}

\subsection{Grammar based compression}

In the grammar-based compression text is represented by a context-free grammar (CFG) generating exactly one string.
The idea behind this approach is that a CFG can compactly represent the structure of the text,
even if this structure is not apparent.
Furthermore, the natural hierarchical definition of the context-free grammars make such a representation suitable for algorithms,
in which case the string operations can be performed on the compressed representation,
without the need of the explicit decompression~\cite{GawryLZ,FCPM,SLPpierwszePM,PlandowskiSLPequivalence,RytterSWAT,SLPaprox2}.
Lastly, there is a close connection between block-based compression methods and the grammar compression:
it is fairly easy to rewrite the LZW definition as a $\Ocomp(1)$ larger CFG,
LZ77 can also be presented in this way, introducing a polynomial blow-up in size
(reducing the blow up to $\log(N/\ell)$, where $\ell$ is the size of the LZ77 representation,
is non-trivial~\cite{SLPaprox,SLPaprox2}).

While grammar-based compression was introduced with practical
purposes in mind and the paradigm was used in several implementations~\cite{RePair,KiefferY96,Sequitur},
it also turned out to be very useful in more theoretical considerations.
Intuitively, in many cases large data have relatively simple inductive definition,
which results in a grammar representation of small size. 
On the other hand, it was already mentioned that the hierarchical structure of the
CFGs allows operations directly on the compressed representation.
A recent survey by Lohrey\cite{Lohreysurvey} gives a comprehensive description
of several areas of theoretical computer science in which grammar-based compression was successfully applied.

The main drawback of the grammar-based compression is that producing the smallest CFG for a text is \emph{intractable}:
given a string $w$ and number $k$ it is {\selectfont\sffamily NP}-hard to decide whether there exist a CFG of size $k$ that generates $w$~\cite{SLPapproxNPhard}.
Furthermore, the size of the grammar cannot be approximated within some small constant factor~\cite{SLPaprox2}.

Lastly, it is worth noting that in an extremely simple cases of texts of the form $a^{\ell_1}ba^{\ell_2}b \cdots b a^{\ell_k}$
construction of the grammar generating this string is equivalent (up to a small constant factor) to a construction of an \emph{addition chain}
for the sequence $\ell_1 < \ell_2 <  \ldots < \ell_k$
and for the latter problem the best algorithm returns an addition chain of size
$\log \ell_k + \Ocomp\Big(\sum_{i=1}^k \frac{\log \ell_i}{\log \log \ell_i}\Big)$~\cite{Yaopowers},
which in particular yields an $\Ocomp\Big( \frac{\log n}{\log \log n} \Big)$ approximation of the size of the smallest addition chain.
Since the addition chains are well studied, showing a construction of an addition chains shorter than
$\log \ell_k + \Ocomp\Big(\sum_{i=1}^k \frac{\log \ell_i}{\log \log \ell_i}\Big)$ seems unlikely.
Still, this construction was not aimed at \emph{approximating} the shortest addition chain,
it is still possible that $\Ocomp(\frac{\log n}{\log \log n})$ approximation can be improved.
In any case, any new result for addition chains would be interesting on its own.

\subsection{Approximation}
The hardness of the smallest grammar problem naturally leads to two directions of research: on one hand, several heuristics are considered~%
\cite{RePair,KiefferY96,Sequitur},
on the other, approximation algorithms, with a guaranteed approximation ratio, are proposed;
in this paper we consider only the latter.

The first two algorithms with an approximation ratio $\Ocomp(\log(N/g))$ were developed simultaneously
by Rytter~\cite{SLPaprox} and Charikar et~al.~\cite{SLPaprox2}.
They followed a similar approach, we first present Rytter's approach as it is a bit easier to explain.

Rytter's algorithm~\cite{SLPaprox} applies the LZ77 compression to the input string and then transforms the
obtained LZ77 representation to a $\Ocomp(\ell \log(N/\ell))$ size grammar,
where $\ell$ is the size of the LZ77 representation.
It is easy to show that $\ell \leq g$ and as $f(x) = x \log(N/x)$ is increasing,
the bound $\Ocomp(g\log(N/g))$ on the size of the grammar follows (and so a bound $\Ocomp(\log(N/g))$ on approximation ratio).
The crucial part of the construction is the requirement that the intermediate constructed grammar
defines a derivation tree satisfying the AVL condition.
The bound on the running time and the approximation guarantee are all consequences of the balanced form of the derivation tree
and of the known algorithms for merging, splitting, etc.\ of AVL trees
(in fact these procedures are much simpler in this case, as we do not store any information in the internal nodes~\cite{SLPaprox}).
Note that also the final grammar for the input text is balanced, which makes is suitable for later processing.
Since the construction of LZ77 representation can be performed in linear time
(assuming that the letters of the input word can be sorted in linear time),
also the running time of the whole algorithm can be easily bounded by a linear function.

Charikar et~al.~\cite{SLPaprox2} followed more or less the same path,
with a different condition imposed on the grammar: it was required that its derivation tree is length-balanced,
i.e.\ for a rule $X \to YZ$ the lengths of words generated by $Y$ and $Z$ are within a certain 
multiplicative constant factor from each other.
For such trees efficient implementation of merging, splitting etc.\ operations were given (i.e.\ constructed from scratch)
by the authors and so the same running time as in the case of the AVL trees was obtained.

Lastly, Sakamoto~\cite{SLPaproxSakamoto} proposed a different algorithm, based on \RePair~\cite{RePair},
which is one of the practically implemented and used algorithms for grammar-based compression.
His algorithm iteratively replaced pairs of different letters and maximal blocks of letters
($a^\ell$ is a \emph{maximal block} if that cannot be extended by $a$ to either side).
A special pairing of the letters was devised, so that it is `synchronising':
if $w$ has $2$ disjoint occurrences in text, then those two occurrences can be represented
as $w_1w'w_2$, where $w_1,w_2 = \Ocomp(1)$, such that both occurrences of $w'$ in text
are paired and compressed in the same way.
The analysis was based on considering the LZ77 representation of the text and proving that due to
`synchronisation' the factors of LZ77 are compressed very similarly as the text to which they refer.

However, to the author's best knowledge and understanding, the presented analysis~\cite{SLPaproxSakamoto}
is incomplete, as the cost of nonterminals introduced when maximal blocks are replaced
is not bounded at all in the paper, see the appendix; the bound that the author was able to obtain using there presented approach
is $\Ocomp(\log(N/g)^2)$, so worse than claimed.

\subsection{Proposed approach: recompression}
In this paper another algorithm is proposed,
it is constructed using the general approach of \emph{recompression}, developed by the author.
In essence, we iteratively apply two replacement schemes to the text $T$:
\begin{description}
	\item[pair compression of $ab$]
	For two different symbols (i.e.\ letters or nonterminals) $a$, $b$ such that substring $ab$ occurs in \mytext{}
	replace each of $ab$ in \mytext{} by a fresh nonterminal $c$.
	\item[$a$'s block compression]
	For each maximal block $a^\ell$, where $a$ is a letter or a nonterminal and $\ell >1$,
	that occurs in \mytext, replace all $a^\ell$s in \mytext{}
	by a fresh nonterminal $a_\ell$.
\end{description}
Then the returned grammar is obtained by backtracking the compression operations performed by the algorithm:
observe that replacing $ab$ with $c$ corresponds to a grammar production
\begin{subequations}
\label{eq: productions}
\begin{equation}
\label{eq: c to ab}
c \to ab
\end{equation}
and similarly
replacing $a^\ell$ with $a_\ell$ corresponds to a grammar production
\begin{equation}
\label{eq: al to al}
a_\ell \to a^\ell 
\enspace .
\end{equation}
\end{subequations}

The algorithm is divided into \emph{phases}: in the beginning of a phase,
all pairs occurring in the current text are listed and stored in a list $P$, similarly,
$L$ contains all letters occurring in the current text.
Then pair compression is applied to an appropriately chosen subset of $P$
and all blocks of symbols from $L$ are compressed, then the phase ends.
In everything works perfectly, each symbol of \mytext{} is replaced and so \mytext's length drops by half;
in reality the text length drops by some smaller, but constant, factor per phase.
For the sake of simplicity, we treat all nonterminals introduced by the algorithm as letters.

In author's previous work it was shown that such an approach can be efficiently applied
to text represented in a grammar compressed form.
In this way new results for compressed membership problem~\cite{fullyNFA},
fully compressed pattern matching~\cite{FCPM} and word equations~\cite{wordequations,onevarlinear} were obtained.
In this paper a somehow opposite direction is followed:
the recompression method is employed to the input string.
This yields a simple linear-time algorithm:
Performing one phase in $\size$ running time is relatively easy,
since the length of \mytext{} drops by a constant factor in each phase,
the $\Ocomp(N)$ running time is obtained.

However, the more interesting is the analysis, and not the algorithm itself:
it is performed by applying (as a mental experiment) the recompression to the optimal grammar $G$ for the input text.
In this way, the current $G$ always generates the current string kept by the algorithm
and the number of nonterminals introduced during the construction can be calculated in terms of $|G| \leq g$.

A relatively straightforward analysis yields that the generated grammar is of size $\Ocomp(g \log N)$,
a slightly more involved algorithm that combines the recompression technique with a naive approach that generates
a grammar of size $\Ocomp(N)$ yields a grammar of size $\Ocomp(g\log(N/g) + g)$.

\subsection{Advantages and disadvantages of the proposed technique}
We believe that the proposed algorithm is interesting, as it is very simple
and its analysis for the first time does not rely on LZ77 representation of the string.
Potentially this can help in both design of an algorithm with a better approximation ratio and in showing a logarithmic lower bound:
Observe that LZ77 representation is known to be at most as large as the smallest grammar,
so it might be that some algorithm produces a grammar of size $o(g \log (N/g))$,
even though this is of size $\Omega(\ell \log(N/\ell))$, where $\ell$ is the size of the LZ77 representation of the string.
Secondly, as the analysis `considers' the optimal grammar, it may be much easier to observe,
where any approximation algorithm performs badly, and so try to approach a logarithmic lower bound.
This is much harder to imagine, when the approximation analysis is done in terms of the LZ77.

Unfortunately, the obtained grammar is not balanced in any sense, in fact it is easy to give examples
on which it returns grammar of height $\Omega(\sqrt N)$ (note though that the same applies also to grammar
returned by Sakamoto's algorithm). This makes the obtained grammar less suitable for later processing;
on the other hand, the practically used grammar-based compressors~\cite{RePair,KiefferY96,Sequitur}
also do not produce a balanced grammar, nor do they give a guarantee on its height.

On the good side, there is no reason why the optimal grammar should be balanced,
neither can we expect that for an unbalanced grammar a small balanced one exists.
Thus it is possible that while $o(\log(N/g))$ approximation algorithm exists,
there is no such algorithm that always returns a balanced grammar.

We note that the reason why the grammar returned by proposed algorithm can have large height
is only due to block compression: if we assume that the nonterminal generating $a^\ell$ has height one,
the whole grammar has height $\Ocomp(\log N)$.
It looks reasonable to assume that many data structures for grammar representation of text as well as later processing of it can
indeed process a production $a_\ell \to a^\ell$ in constant time.

Lastly, the proposed method seems to much easier to generalise then the LZ77-based ones:
generalisations of SLPs to grammars generating other objects (mostly: trees)
are known but it seems that LZ77-based approach does not generalise to such setting,
as LZ77 ignores any additional structure (like: tree-structure) of the data.
In recent work of Lohrey and the author the algorithm presented in this paper
is generalised to the case of tree-grammars, yielding a first provable approximation
for the smallest tree grammar problem~\cite{treegrammar}.

\subsubsection*{Comparison with Sakamoto's algorithm}
The general approach is similar to Sakamoto's method,
however, the pairing of letters seems more natural in here presented paper.
Also, the construction of nonterminals for blocks of letters is different,
the author failed to show that the bound actually holds for the variant proposed by Sakamoto.
It should be noted that the analysis presented in this paper for the calculation of nonterminals
used due to pair compression is fairly easy, while estimating the number used for block compression
is much more involved. Also, the connection to the addition chains suggests that the compression of blocks
is the difficult part of the smallest grammar problem.

\subsubsection*{Note on computational model}
The presented algorithm runs in linear time, assuming that the $\Sigma$ can be identified with a continues
subset of natural numbers of size $\Ocomp(N^c)$ for some constant $c$ and the \algradix{} can be performed on it.
Should this not be the case for the input, we can replace the original letters with such a subset,
in $\Ocomp(n \log |\Sigma|)$ time (by creating a balanced tree for letters occurring in the input string).
Note that the same comment applies to previous algorithms: there are many different algorithms for constructing
the LZ77 representation of the text, but all of them first compute a suffix array (or a suffix tree) of the text,
and linear-time algorithms for that are based on linear-time sorting of letters (treated as integers);
although Sakamoto's method was designed to work with constant-size alphabet,
it can be easily extended to the case when $\Sigma$ can be identified with a sequence of $\Ocomp(N^c)$ numbers,
retaining the linear running-time.

\section{The algorithm}
The input sequence to be represented by a context-free grammar is $T \in \Sigma^*$ and $N$ denotes its initial length.
The algorithm \algmain{} introduces new symbols to the instance, which are the nonterminals of the constructed grammar.
However, these are later treated exactly as the original letters, so we insist on calling them letters as well
and use common set $\Sigma$ for both letters and nonterminals.
We assume that $T$ is represented as a doubly-linked list, so that removal and replacement of its elements can be performed
in constant time (assuming that we have a link to such an occurrence).
Note though that if we were to store $T$ in a table, the running time would be the same.

The smallest grammar generating $T$ is denoted by $G$ and its size $|G|$, measured as the length of the productions, is $g$.
The crucial part of the analysis is the modification of $G$ according to the compression performed on $T$.
The terms nonterminal, rules, etc.\ always address the optimal grammar $G$ (or its transformed version).
To avoid confusion, we do not use terms `production' and `nonterminal' for $a$ that replaced some substring in \mytext{}
(even though this is formally a nonterminal of the constructed grammar).
Still, when a new `letter' $a$ is introduced to \mytext{}
we need to estimate the length of the `productions' in the constructed grammar that are needed for $a$
(note that we can of course use all letters previously used in \mytext).
We refer to this length of productions as \emph{cost of representation of a letter} $a$.
For example, in production~\eqref{eq: c to ab} then the representation cost is $2$
(as we have only one rule $c \to ab$, this rule is called a \emph{representation} of $c$)
and in a rule~\eqref{eq: al to al} we have a cost $\ell$;
the latter cost can be significantly reduced, for instance for $a^12$ we can have a representation cost of $8$ instead of $12$,
when we use a subgrammar $a_2 \to aa$, $a_3 \to a_2a$, $a_6 \to a_3 a_3$ and $a_12 \to a_6a_6$.
%
%
%
%
Note that when $c$ replaces a pair (as in~\eqref{eq: c to ab}), its representation cost is always $2$,
but when $a$ replaces a block of letters, say $a^\ell$, the cost might be larger than constant.
In the latter case our algorithm constructs a special subgrammar for $a_\ell$ that generates $a^\ell$.
Details are explained later on.

\begin{algorithm}[H]
	\caption{\algmain: outline}
	\label{alg:main}
	\begin{algorithmic}[1]
	\While{$|\mytext|>1$} \label{alg:mainloop}
		\State $L \gets $ list of letters in \mytext{} \label{listing letters}
    \For{each $a \in L$} \Comment{Blocks compression}
    		\State compress maximal blocks of $a$ \label{outer occurrences} \Comment{\size}
    \EndFor				    
		\State $P \gets $ list of pairs \label{listing pairs}
		\State find partition of $\Sigma$ into $\Sigma_\ell$ and $\Sigma_r$ \label{partition letters}
			\Comment{Covering at least $1/2$ of occurrences of letters in \mytext}
			\\ \Comment{\size, see Lemma~\ref{lem:finding partition}}
    \For{$ab\in P \cap \Sigma_\ell\Sigma_r$} \Comment{These pairs do not overlap}
    	\State compress pair $ab$ \label{pair compression} \Comment{Pair compression}
    \EndFor
	\EndWhile
	\State \Return the constructed grammar
	 \end{algorithmic}
\end{algorithm}

We call one iteration of the main loop of \algmain{} a \emph{phase}.

Before we make any analysis, we note that at the beginning of each phase
we can make a linear-time preprocessing that guarantees that the letters in \mytext{}
form an interval of numbers (which makes them more suitable for sorting using \algradix).

\begin{lemma}
\label{lem:letters are consecutive}
At the beginning of the phase, in time \size{} we can rename the letters used in \mytext{}
so that they form an interval of numbers.
\end{lemma}
\begin{proof}
Observe that we assumed that the input alphabet consists of letters that can be identified with subset of $\{1, \ldots, N^c\}$,
see the discussion in the introduction.
Treating them as vectors of length $c$ over $\{0, \ldots, N - 1\}$ we can sort them using \algradix{}
in $\Ocomp(cN)$ time, i.e.\ linear one. Then we can re-number those letters to $1$, $2$, \ldots, $n$ for some $n \leq N$.

Suppose that at the beginning of the phase the letters form an interval $[m,\twodots, m + k]$.
Each new letter, introduced in place of a compressed subpattern (i.e.\ a block $a^\ell$ or a pair $ab$),
is assigned a consecutive value, and so after the phase the letters appearing in  \mytext{} are within an interval $[m\twodots m + k']$
for some $k' > k$. It is now left to re-number the letters from $[m\twodots m + k']$,
so that the ones appearing in \mytext{} indeed form an interval.
For each symbol $a$ in the interval $[m\twodots m + k']$ we set a flag to $\flag[a] = 0$. Moreover, we set a variable $\next$ to $m+k'+1$. 
Then we read \mytext.
Whenever we spot a letter $a \in [m\twodots m + k']$ with $\flag[a] = 0$, we set
$\flag[a] := 1$; $\mynew[a] := \next$, and $\next := \next+1$.
Moreover, we replace this $a$ by $\mynew[a]$.
When we spot a symbol $a \in [m\twodots m + k']$ with $\flag[a] = 1$, then
we replace this $a$ by $\mynew[a]$.
Clearly the running time is \size{} and after the algorithm the symbols form a subinterval of $[m+k'+1\twodots m+2k'+1]$.
\qedhere
\end{proof}

\subsection{Blocks compression}
The blocks compression is very simple to implement:
We read \mytext, for a maximal block of $a$s of length greater than $1$
we create a record $(a,\ell,p)$, where $\ell$ is a length of the block,
and $p$ is the pointer to the first letter in this block.
We then sort these records lexicographically using \algradix{} (ignoring the last component).
There are only $\size$ records and we assume that $\Sigma$ can be identified with an interval,
see Lemma~\ref{lem:letters are consecutive}, this is all done in \size.
Now, for a fixed letter $a$, the consecutive tuples with the first coordinate $a$
correspond to all blocks of $a$, ordered by the size.
It is easy to replace them in \size{} time with new letters.

Note that so far we did not care with the cost of representation of new letters that replaced $a$-blocks.
We use a particular schema to represent $a_{\ell_1}, a_{\ell_2}, \ldots, a_{\ell_k}$, which shall have a representation cost
$\Ocomp(\sum_{i=1}^k[1+ \log(\ell_{i} - \ell_{i-1})])$ (take $\ell _0 = 0$ for convenience).

\begin{lemma}
	\label{lem: cost of powers}
Given a list $1 < \ell_1 <  \ell_2 < \dots < \ell_k$ we can represent letters $a_{\ell_1}, a_{\ell_2}, \ldots, a_{\ell_k}$
that replace blocks $a^{\ell_1}, a^{\ell_2}, \ldots, a^{\ell_k}$ with a cost
$ \Ocomp(\sum_{i=1}^k[1 + \log(\ell_{i} - \ell_{i-1})])$, where $\ell_0 = 0$.
\end{lemma}
\begin{proof}
Firstly observe that without loss of generality we may assume that
the list $\ell_1, \ell_2, \dots, \ell_k$ is given to us in a sorted way,
as it can be easily obtained from the sorted list of occurrences of blocks.
For simplicity define $\ell_0 = 0$ and let $\ell = \max_{i=1}^k (\ell_{i} - \ell_{i-1})$.

In the following, we shall define rules for certain new letters $a_m$, each of them `derives' $a^m$
(in other words, $a_m$ represents $a^m$).
For each $1 \leq i \leq \log \ell$ introduce a new letter $a_{2^i}$, defined as $a_{2^i} \to a_{2^{i-1}}a_{2^{i-1}}$,
where $a_1$ simply denotes $a$. Clearly $a_{2^i}$ represents $a^{2^i}$ and the representation
cost summed over all $i\leq \ell$ is $\Ocomp(\log \ell)$.

Now introduce new letters $a_{\ell_i - \ell_{i-1}}$ for each $i > 0$,
which shall represent $a^{\ell_i - \ell_{i-1}}$.
They are represented using the binary expansion, i.e.\ by concatenation of at most $1 + \log(\ell_i - \ell_{i-1})$
from among the letters $a_1, a_2, a_4, \ldots, 2^{ \lfloor \log(\ell_i - \ell_{i-1}) \rfloor}$.
This has a representation cost $\Ocomp(\sum_{i=1}^k [1 + \log(\ell_{i} - \ell_{i-1})])$.

Lastly, each $a_{\ell_i}$ is  represented as $a_{\ell_{i}} \to a_{\ell_{i} -\ell_{i-1}} a_{\ell_{i-1}}$,
which has a total representation cost $\Ocomp(k)$.

Summing up $\Ocomp(\log \ell)$, $\Ocomp(\sum_{i=1}^k [1 + \log(\ell_{i} - \ell_{i-1})])$ and $\Ocomp(k)$
we obtain $\Ocomp(\sum_{i=1}^k [1 + \log(\ell_{i} - \ell_{i-1})])$, as claimed.
\end{proof}

In the following we shall also use a simple property of the block compression:
since no two maximal blocks of the same letter can be next to each other,
after the block compression there are no blocks of length greater than $1$ in \mytext.

\begin{lemma}
\label{lem:consecutive letters are different}
In line~\ref{listing pairs} there are no two consecutive letters $aa$ in \mytext.
\end{lemma}
\begin{proof}
Suppose for the sake of contradiction that there are such two letters.
There are two cases:

\begin{description}
	\item[$a$ was present in \mytext{} in line~\ref{listing letters}]
But then $a$ was listed in $L$ in line~\ref{listing letters} and $aa$
was replaced by another letter in line~\ref{outer occurrences}.

\item[$a$ was introduced in line~\ref{outer occurrences}]
Both $a$ replaced some maximal blocks $b^\ell$
thus $aa$ replaced $b^{2\ell}$, and so each of those two $b^\ell$s was not a maximal block.
\qedhere
\end{description}
\end{proof}

\subsection{Pair compression}
The pair compression is performed similarly as the block compression.
However, since the pairs can overlap, compressing all pairs at the same time is not possible.
Still, we can find a subset of non-overlapping pairs in \mytext{} such that a constant fraction
of letters \mytext{} is covered by occurrences of these pairs.
This subset is defined by a \emph{partition} of $\Sigma$ into $\Sigma_\ell$ and $\Sigma_r$
and choosing the pairs with the first letter in $\Sigma_\ell$ and the second in $\Sigma_r$.

\begin{lemma}
\label{lem:finding partition}
For \mytext{} in \size{} time we can find in line~\ref{partition letters}
a partition of $\Sigma$ into $\Sigma_\ell$, $\Sigma_r$ such that
number of occurrences of pairs $ab \in \Sigma_\ell\Sigma_r$ in \mytext{} is at least $(|\mytext|-1)/4$.

In the same running time we can provide, for each $ab \in P \cap \Sigma_\ell\Sigma_r$,
a lists of pointers to occurrences of $ab$ in \mytext.
\end{lemma}
\begin{proof}
For a choice of $\Sigma_\ell \Sigma_r$ we say that occurrences of $ab \in P \cap \Sigma_\ell\Sigma_r$
are \emph{covered} by $\Sigma_\ell \Sigma_r$.


The existence of partition covering at least one fourth of the occurrences
can be shown by a simple probabilistic argument:
divide $\Sigma$ into $\Sigma_\ell$ and $\Sigma_r$ randomly, where each letter
goes to each of the parts with probability $1/2$.
Consider two consecutive letters $ab$ in \mytext,
note that they are different by Lemma~\ref{lem:consecutive letters are different}.
Then $a \in \Sigma_\ell$ and $b \in \Sigma_r$ with probability $1/4$.
There are $|\mytext|-1$ such pairs in \mytext, so the expected number of pairs in \mytext{}
from $\Sigma_\ell\Sigma_r$ is $(|\mytext|-1)/4$.
Observe, that if we were to count the number of pairs that are covered \emph{either} by $\Sigma_\ell\Sigma_r$
or by $\Sigma_r\Sigma_\ell$ then the expected number of pairs covered by $\Sigma_\ell \Sigma_r \cup \Sigma_r \Sigma_\ell$
is $(|\mytext|-1)/2$.

The deterministic construction of such a partition follows by a simple derandomisation,
using an expected value approach.
It is easier to first find a partition such that at least $(|\mytext|-1)/2$ pairs' occurrences in \mytext{} are covered
by $\Sigma_\ell\Sigma_r \cup \Sigma_r\Sigma_\ell$, we then choose $\Sigma_\ell\Sigma_r$ or $\Sigma_r\Sigma_\ell$,
depending on which of them covers more occurrences.

Suppose that we have already assigned some letters to $\Sigma_\ell$ and $\Sigma_r$
and we are to decide, where the next letter $a$ is assigned.
If it is assigned to $\Sigma_\ell$, then all occurrences of pairs from $a\Sigma_\ell \cup \Sigma_\ell a$ are not going
to be covered, while occurrences of pairs from $a \Sigma_r \cup \Sigma_r a$ are;
similarly observation holds for $a$ being assigned to $\Sigma_r$.
The algorithm makes a greedy choice, maximising the number of covered pairs in each step.
As there are only two options, the choice brings in at least half of occurrences considered.
Lastly, as each occurrence of a pair $ab$ from \mytext{} is considered exactly once
(i.e.\ when the second of $a$, $b$ is considered in the main loop),
this procedure guarantees that at least half of occurrences of pairs in \mytext{} is covered.

In order to make the selection effective,
the algorithm \alggreedypairs{} keeps an up to date counters
$\mycount_\ell[a]$ and $\mycount_r[a]$, 
denoting, respectively, the number of occurrences of pairs from
$a\Sigma_\ell \cup \Sigma_\ell a$ and $a\Sigma_r \cup \Sigma_r a$ in \mytext.
Those counters are updated as soon as a letter is assigned to $\Sigma_\ell$ or $\Sigma_r$.

\begin{algorithm}[H]
  \caption{\alggreedypairs{} \label{greedy pairs}}
  \begin{algorithmic}[1]
  \State $L \gets $ set of letters used in $P$
	\State $\Sigma_\ell \gets \Sigma_r \gets \emptyset$ \Comment{Organised as a bit vector}
	\For{$a \in L$}
		\State $\mycount_\ell[a] \gets \mycount_r[a] \gets 0$ \Comment{Initialisation}
	\EndFor
	\For{$a \in L$}
		\If{$\mycount_r[a] \geq \mycount_\ell[a]$} \Comment{Choose the one that guarantees larger cover}
			\State $\textnormal{\textsl{choice}} \gets \ell$ 
		\Else
			\State $\textnormal{\textsl{choice}} \gets r$
		\EndIf
		\State $\Sigma_{\textnormal{\textsl{choice}}} \gets \Sigma_{\textnormal{\textsl{choice}}} \cup \{a\}$\
		\For{each $ab$ or $ba$ occurrence in \mytext{}}
			\State $\mycount_{\textnormal{\textsl{choice}}}[b] \gets \mycount_{\textnormal{\textsl{choice}}}[b] + 1$ \label{counter update}
		\EndFor
	\EndFor
	\If{\# occurrences of pairs from $\Sigma_r\Sigma_\ell$ in \mytext > \# occurrences of pairs from $\Sigma_\ell\Sigma_r$ in \mytext}
	\label{actual partition choice}
		\State switch $\Sigma_r$ and $\Sigma_\ell$ \label{actual partition}
	\EndIf
	\State \Return $(\Sigma_\ell, \Sigma_r)$
  \end{algorithmic}
\end{algorithm}

By the argument given above, when $\Sigma$ is partitioned into $\Sigma_\ell$ and $\Sigma_r$ by \alggreedypairs,
at least half of the occurrences of pairs from \mytext{} are covered by
$\Sigma_\ell\Sigma_r \cup \Sigma_r \Sigma_\ell$.
Then one of the choices $\Sigma_\ell\Sigma_r$ or $\Sigma_r\Sigma_\ell$ covers at least one fourth of the occurrences.

It is left to give an efficient variant of \alggreedypairs,
the non-obvious operations are the choice of the actual partition in lines~\ref{actual partition choice}--\ref{actual partition}
and the updating of $\mycount_\ell[b]$ or $\mycount_r[b]$ in line~\ref{counter update}. 
All other operation clearly take at most \size{} time.
The former is simple: since we organise $\Sigma_\ell$ and $\Sigma_r$ as bit vectors,
we can read \mytext{} from left to right and calculate the number of pairs from $\Sigma_\ell\Sigma_r$ and those from $\Sigma_r\Sigma_\ell$
in \size{} time (when we read a pair $ab$ we check in $\Ocomp(1)$ time whether $ab \in \Sigma_\ell \Sigma_r$ or $ab \in \Sigma_r \Sigma_\ell$).
Afterwards we choose the partition that covers more occurrences of pairs in \mytext.

To implement the $\mycount$,
for each letter $a$ in \mytext{} we have a \emph{right list}
$\rightl = \makeset{b}{ab \text{ occurs in } \mytext}$, represented as a list.
Furthermore, the element $b$ on right list stores
a list of all occurrences of the pair $ab$ in \mytext.
There is a similar \emph{left list} $\leftl = \makeset{b}{ba \text{ occurs in } \mytext}$.
We comment, how to create left lists and right lists later.

Given \rightl[] and \leftl[], performing the update in line~\ref{counter update} is easy:
we go through $\rightl$ ($\leftl$) and increase the $\mycount_{\textsl{choice}}[b]$ for each occurrence of $ab$ ($ba$, respectively).
Note that in this way each of the list $\rightl$ ($\leftl$) is read once during \alggreedypairs,
and so this time can be charged to their creation.

It remains to show how to initially create \rightl{} (\leftl{} is created similarly).
We read \mytext, when reading a pair $ab$ we create a record
$(a,b,p)$, where $p$ is a pointer to this occurrence.
We then sort these record lexicographically using \algradix.
There are only $\size$ records and we assume that $\Sigma$ can be identified with an interval,
see Lemma~\ref{lem:letters are consecutive}, this is all done in \size.
Now, for a fixed letters $a$, the consecutive tuples with the first coordinate $a$
can be turned into \rightl: for $b \in \rightl$ we want to store a list $I$ of pointers to occurrences of $ab$,
and on a sorted list of tuples the $\{ (a,b,p) \}_{p \in I}$ are consecutive elements.

Lastly, in order to get for each $ab \in P \cap \Sigma_\ell\Sigma_r$,
the lists of pointers to occurrences of $ab$ in \mytext{}
it is enough to read \rightl[] and filter the pairs such that $a\in \Sigma_\ell$ and $b \in \Sigma_r$;
the filtering can be done in $\Ocomp(1)$ as $\Sigma_\ell$ and $\Sigma_r$ are represented as bitvectors.
The needed time is \size.

The total running time is also \size, as each subprocedure has time constant per pair processed or \size{} in total.
\end{proof}

When for each pair $ab \in \Sigma_\ell\Sigma_r$ the list of its occurrences in \mytext{} is provided,
the replacement of pairs is done by going through the list and replacing each of the pair, which is done in linear time.
Note, that as $\Sigma_\ell$, $\Sigma_r$ are disjoint, the considered pairs cannot overlap.

\subsection{Size and running time}
It remains to estimate the total running time, summed over all phases.
Clearly each subprocedure in a phase has a running time \size{}
so it is enough to show that $|T|$ is reduced by a constant factor per phase.
\begin{lemma}
\label{lem:number of phases}
In each phase $|\mytext|$ is reduced by a constant factor.
\end{lemma}
\begin{proof}
Let $m = |\mytext|$ at the beginning of the phase.
Let $m' \leq m$ be the length of \mytext{} after the compression of blocks.
First observe that if $m < 5$ then we satisfy the lemma when we make at least one compression, which can be always done,
so in the following we assume that $m \geq 5$.

By Lemma~\ref{lem:finding partition} at least $(m'-1)/4$ pairs are compressed during the pair compression,
hence after this phase $|\mytext'| \leq m' - (m'-1)/4 \leq \frac 3 4 m +\frac 1 4$.
\end{proof}

\begin{theorem}
\algmain{} runs in linear time.
\end{theorem}
\begin{proof}
Each phase clearly takes \size{} time and by Lemma~\ref{lem:number of phases}
the $|\mytext|$ drops by a constant factor in each phase.
As the initial length of $\mytext$ is $N$, the total running time is $\Ocomp(N)$.
\end{proof}

\section{Size of the grammar: SLPs and recompression}
To bound cost of representing the letters introduced during the construction of the grammar,
we start with the smallest grammar $G$ generating (the input) \mytext{}
and then modify it so that it generates \mytext{} (i.e.\ the current string kept by \algmain) after each of the compression steps.
Then the cost of representing the introduced letters is paid by various credits assigned to $G$.
Hence, instead of the actual representation cost, which is difficult to estimate, we calculate the total value of issued credit.
Note that this is entirely a mental experiment for the purpose of the analysis,
as $G$ is not stored or even known to the algorithm.
We just perform some changes on it depending on the \algmain{} actions.

We assume that grammar $G$ is a \emph{Straight Line Programme} (\emph{SLP}),
however, we relax the notion a bit (and call it an \emph{SLP with explicit letters}, when an explicit reference is needed):
i.e.\ its nonterminals are numbered $X_1$, \ldots, $X_m$ and each rule has at most
two nonterminals (with smaller indices) in its body, 
(i.e.\ there are two, one or none nonterminals and arbitrary number of letters in the rule's body).
Note that every CFG generating a unique string can be transformed into an SLP with explicit letters,
with the size increased only by a constant factor.
We call the letters (strings) occurring in the productions the \emph{explicit letters} (\emph{strings}, respectively).
The unique string derived by $X_i$ is denoted by $\eval(X_i)$;
the grammar $G$ shall satisfy the condition $\eval(X_m) = \mytext$.
We do not assume that $\eval(X_i) \neq \epsilon$, however, if $\eval(X_i) = \epsilon$ then $X_i$
is not used in the productions of $G$ (as this is a mental experiment, such $X_i$ can be removed from the rules
and in fact from the SLP).

With each explicit letter we associate two units of \emph{credit} and pay 
most of the cost of representing the letters introduced during \algmain{} with these credits.
More formally:
when the algorithm modifies $G$ and in the process it creates an occurrence of a letter,
we \emph{issue} (or pay) $2$ new credits.
On the other hand, if we do a compression step in $G$, then we remove some
occurrences of letters. The credit associated with these occurrences is then \emph{released}
and can be used to pay for the representation cost of the new letters introduced by the compression step
(so that the algorithm does not issue new credit). 
For pair compression the released credit indeed suffices to pay
both the credit of the new letters occurrences and their representation cost,
but for chain compression the released credit does not suffice,
as it is not enough to pay the representation cost.
Here we need some extra amount that will be estimate separately later on
in Section~\ref{subsec: blocks representation cost}.
In the end, the total cost is the sum of credit that was issued during the modifications of $G$
plus the value that we estimate separately in Section~\ref{subsec: blocks representation cost}.

Recall that whenever we say nonterminal, rule, production etc., we mean one of $G$.

\subsection{Intuition}
When we replace each occurrence of the pair $ab$ in \mytext,
we should also do this in $G$. However, this may be not possible, as some $ab$
generated by $G$ do not come from explicit pairs in $G$ but rather are `between' a nonterminal
and a letter, for instance in a simple grammar $X_1 \to a$, $X_2 \to X_1b$ the pair $ab$ has such a problematic occurrence.
If there are no such occurrences, it is enough to replace each explicit $ab$ in $G$ and we are done.
To deal with the problematic ones, we need to somehow change the grammar,
in the example above we replace $X_1$ with $a$, leaving only $X_2 \to ab$, for which the previous procedure can be applied.
It turns out that a systematic procedure that deals with all such problems at once can be given,
it is the main ingredient of this section and it is given in Section~\ref{subsec: pair compression}.
Similar problems occur also when we want to replace maximal blocks of $a$ and the solution
to this problem is similar and it is given in Section~\ref{subsec:block compression}.

Note that in the example above, when $X_1$ is replaced with $a$, $2$ credit for the occurrence of $a$ in $X_1 \to a$ is released and wasted.
Then we issue $2$ credit for the new occurrence of $a$ in the rule $X_2$.
When $ab$ is replaced with $c$, $4$ credit is released when $ab$ is removed from the rule,
$2$ of this credit is used for the credit of $c$ and the remaining $2$ can be used to pay the representation cost for $c \to ab$.

\subsection{Pair compression}
\label{subsec: pair compression}
A pair of letters $ab$ has a \emph{crossing occurrence} in a nonterminal $X_i$
(with a rule $X_i \to \alpha_i$) if $ab$ is in $\eval(X_i)$ but this occurrence does not
come from an explicit occurrence of $ab$ in $\alpha_i$ nor it is generated by any of the
nonterminals in $\alpha_i$.
A pair is \emph{non-crossing} if it has no crossing occurrence.
Unless explicitly written, we use this notion only to pairs of \emph{different} letters.

By $\PC_{ab\to c}(w)$ we denote the text obtained from $w$ by replacing each $ab$ by a letter $c$
(we assume that $a \neq b$).
We say that a procedure (that changes a grammar $G$ with nonterminals $X_1, \ldots, X_m$
to $G'$ with nonterminals $X_1', \ldots, X_m'$)
\emph{properly implements the pair compression} of $ab$ to $c$,
if $\eval(X_m') = \PC_{ab \to c}(\eval(X_m))$ and $G'$ is an SLP with explicit letters.
When a pair $ab$ is noncrossing the procedure that implements the pair compression
is easy to give: it is enough to replace each explicit $ab$ with $c$.

\begin{algorithm}[H]
  \caption{$\algpairncr(ab,c)$: compressing a non-crossing pair $ab$.
  \label{alg:noncrossing compression}}
  \begin{algorithmic}[1]
	\State replace each explicit $ab$ in $G$ by $c$
  \end{algorithmic}
\end{algorithm}

In order to distinguish between the nonterminals, grammar, etc.\
before and after the application of compression of $ab$ (or, in general, any procedure) 
we use `primed' letters, i.e.\  $X_i'$, $G'$, \textci{}
for the nonterminals, grammar and text after this compression and `unprimed',
i.e.\ $X_i$, $G$, \mytext{} for the ones before.

\begin{lemma}
\label{lem:noncrossing compression}
If $ab$ is a noncrossing pair, then $\algpairncr(ab,c)$ properly implements the compression of $ab$.
The credit of new letters in $G'$ and cost of representing the new letter $c$ is paid by the released credit;
no new credit is issued.
If a pair $de$, where $d \neq c \neq e$ is noncrossing in $G$, it is in $G'$.
\end{lemma}
\begin{proof}
By induction on $i$ we show that $\eval(X_i') = \PC_{ab \to c}(\eval(X_i))$.
Consider any occurrence of $ab$ in the string generated by $X_i$.
If it is an explicit string then it is replaced by $\algpairncr(ab,c)$.
If it is contained within substring generated by some $X_j$, this occurrence was compressed
by the inductive assumption.
The remaining case is the crossing occurrence of $ab$:
since the only modifications to the rules made by $\algpairncr(ab,c)$
is the replacement of $ab$ by $c$, such a crossing pair existed already before $\algpairncr(ab,c)$,
but this is not possible by the lemma assumption that $ab$ is non-crossing.

Each occurrence of $ab$ had two units of credit while $c$ has only $2$,
so the replacement released $4$ units of credit,
$2$ of which are used to pay for the credit of $c$ and the other $2$ to pay the cost of representation of $c$
(if we replace more than one occurrence of $ab$, some credit is wasted).

Lastly, replacing $ab$ in $G$ by a new letter $c$ cannot make $de$ (where $d \neq c \neq e$) a crossing pair in $G$,
as no new occurrence of $d$, $e$ was introduced on the way.
\qedhere
\end{proof}

If all pairs in $\Sigma_\ell\Sigma_r$ are non-crossing, iteration of $\algpairncr(ab,c)$ for each pair
$ab$ in $\Sigma_\ell\Sigma_r$ properly implements the pair compression for all pairs in $\Sigma_\ell\Sigma_r$
(note that as $\Sigma_\ell$ and $\Sigma_r$ are disjoint, occurrences of different pairs from $\Sigma_\ell\Sigma_r$
cannot overlap and so the order of replacement does not matter).
So it is left to assure that indeed the pairs from $\Sigma_\ell\Sigma_r$ are all noncrossing.
It is easy to see that $ab \in \Sigma_\ell\Sigma_r$ is a crossing pair if and only if
one of the following three `bad' situations occurs:
\begin{enumerate}[CP1]
	\item \label{CP1}there is a nonterminal $X_i$, where $i<m$, such that $\eval(X_i)$ begins with $b$
	and $aX_i$ occurs in one of the rules;
	\item \label{CP2}there is a nonterminal $X_i$, where $i<m$, such that $\eval(X_i)$ ends with $a$
	and $X_ib$ occurs in one of the rules;
	\item \label{CP3}there are nonterminals $X_i$, $X_j$, where $i,j <m$, such that $\eval(X_i)$ ends with $a$
	and $\eval(X_j)$ begins with $b$ and $X_iX_j$ occurs in one of the rules.
\end{enumerate}
Consider \CPref{1}, let $bw = \eval(X_i)$.  Then it is enough to modify the rule for $X_i$
so that $\eval(X_i) = w$ and replace each $X_i$ in the rules by $bX_i$,
we call this action the \emph{left-popping $b$ from $X_i$}.
Similar operation of right-popping a letter $a$ from $X_i$ is symmetrically defined.
It is shown in the Lemma~\ref{lem:uncrossing pairs} below that they indeed take care of all crossing occurrences of $ab$.

Furthermore, left-popping and right-popping can be performed for many letters in parallel:
the below procedure $\algpop(\Sigma_\ell,\Sigma_r)$ `uncrosses' all pairs
from the set $\Sigma_\ell\Sigma_r$, assuming that $\Sigma_\ell$ and $\Sigma_r$ are disjoint
subsets of $\Sigma$ (and we apply it only in the cases in which they are).

\begin{algorithm}[H]
  \caption{$\algpop(\Sigma_\ell,\Sigma_r)$: Popping letters from $\Sigma_\ell$ and $\Sigma_r$
  \label{pop code full}}
  \begin{algorithmic}[1]
  \For{$i \gets 1 \twodots m-1$}
		\State let the production for $X_i$ be $X_i \to \alpha_i$
		\If{the first symbol of $\alpha_i$ is $b \in \Sigma_r$} \label{still to the right}  \Comment{Left-popping $b$}
			\State remove this $b$ from $\alpha_i$
			\State replace $X_i$ in $G$'s productions by $bX_i$
			\If{$\eval(X_i) = \epsilon$}
				\State remove $X_i$ from $G$'s productions
			\EndIf
		\EndIf
	\EndFor	
  \For{$i \gets 1 \twodots m-1$}
		\State let the production of $X_i$ be $X_i \to \alpha_i$
		\If{the last symbol of $\alpha_i$ is $a \in \Sigma_\ell$} \label{still to the left} \Comment{Right-popping $a$}
			\State remove this $a$ from $\alpha_i$
			\State replace $X_i$ in $G$'s productions by $X_ia$ 
			\If{$\eval(X_i) = \epsilon$}
				\State remove $X_i$ from $G$'s productions
			\EndIf
		\EndIf
	\EndFor
  \end{algorithmic}
\end{algorithm}

\begin{lemma}
\label{lem:uncrossing pairs}
After application of $\algpop(\Sigma_\ell,\Sigma_r)$, where $\Sigma_\ell \cap \Sigma_r = \emptyset$,
none of the pairs $ab\in \Sigma_\ell\Sigma_r$ is crossing.
Furthermore, $\eval(X_m') = \eval(X_m)$.
At most $\Ocomp(m)$ credit is issued during $\algpop(\Sigma_\ell,\Sigma_r)$.
\end{lemma}
\begin{proof}
Observe first that whenever we remove $b$ from the front of some $\alpha_i$ we replace each of $X_i$ occurrence with $bX_i$
and if afterwards $\eval(X_i) = \epsilon$ then we remove $X_i$ from the rules,
hence the words derived by each other nonterminal (in particular $X_m$) do not change,
the same applies to replacement of $X_i$ with $X_ia$.
Hence, in the end $\eval(X_m') = \eval(X_m) = \mytext$ (note that we do not pop letters from $X_m$).

Secondly, we show that if $\eval(X_i')$ begins with a letter $b' \in \Sigma_r$ then we left-popped a letter from $X_i$
(which by the code is some $b \in \Sigma_r$),
a similar claim (by symmetry) of course holds for the last letter of $\eval(X_i)$ and $\Sigma_r$.
So suppose that the claim is not true and consider the nonterminal $X_i$ with the smallest $i$ such that $\eval(X_i')$ begins with $b' \in \Sigma_r$
but we did not left-pop a letter from $X_i$.
Consider what was the first symbol in $\alpha_i$ when \algpop{} considered $X_i$
in line~\ref{still to the right}. As \algpop{} did not left-pop a letter from $X_i$, the first letter of $\eval(X_i)$ and
$\eval(X_i')$ is the same and hence it is $b' \in \Sigma_r$.
So $\alpha_i$ cannot begin with a letter as then it is $b' \in \Sigma_r$, which should have been left-popped.
Hence it is some nonterminal $X_j$ for $j < i$.
But then $\eval(X_j')$ begins with $b' \in \Sigma_r$ and so by the induction assumption \algpop{} left-popped a letter from $X_j$.
But there was no way to remove this letter from $\alpha_i$, so $\alpha_i$ should begin with a letter, contradiction.

Suppose that after \algpop{} there is a crossing pair $ab \in \Sigma_\ell\Sigma_r$.
There are three already mentioned cases \CPref{1}--\CPref{3}: consider only \CPref{1},
in which $aX_i$ occurs in the rule and $\eval(X_i)$ begins with $b$.
Note that as $a \notin \Sigma_r$ is the letter to the left of $X_i'$,
$X_i'$ did not left-pop a letter.
But it begins with $b \in \Sigma_r$, so it should have. Contradiction.
The other cases are dealt with in a similar manner.

Note that at most $4$ new letters are introduced to each rule, thus at most $8m$ credit is issued.
\end{proof}

In order to compress pairs from $\Sigma_\ell\Sigma_r$ it is enough to first uncross them all using $\algpop(\Sigma_\ell,\Sigma_r)$ 
and then compress them all by $\algpairncr(ab,c)$ for each $ab \in \Sigma_\ell\Sigma_r$.

\begin{algorithm}[H] \normalsize
  \caption{$\algpaircr(\Sigma_\ell,\Sigma_r)$: compresses pairs from $\Sigma_\ell \Sigma_r$}
  \begin{algorithmic}[1]
	\State run $\algpop(\Sigma_\ell,\Sigma_r)$
	\For{$ab \in \Sigma_\ell\Sigma_r$}
		\State run $\algpairncr(ab,c)$ \Comment{$c$ is a fresh letter}
	\EndFor
  \end{algorithmic}
\end{algorithm}

\begin{lemma}
\label{lem:pc crossing}
\algpaircr{} implements pair compression for each $ab \in \Sigma_\ell\Sigma_r$.
It issues $\Ocomp(m)$ new credit to $G$, where $m$ is the number of nonterminals of $G$.
The credit of the new letters introduced to $G$ and their representation costs
are covered by the credit issued or released by \algpaircr.
\end{lemma}
\begin{proof}
By Lemma~\ref{lem:uncrossing pairs} after $\algpop(\Sigma_\ell,\Sigma_r)$ each pair in $\Sigma_\ell\Sigma_r$ is non-crossing and 
$\Ocomp(m)$ credit is issued in the process, furthermore $\eval(X_m)$ does not change.

By Lemma~\ref{lem:noncrossing compression} for a non-crossing pair $ab$ the $\algpairncr(ab,c)$
implements the pair compression, furthermore, any other non-crossing pair $a'b' \in \Sigma_\ell\Sigma_r$
remains non-crossing. Lastly, all occurrences of different pairs from $\Sigma_\ell\Sigma_r$ are disjoint (as $\Sigma_\ell$ and $\Sigma_r$
are disjoint subsets of $\Sigma$) as so the order of replacing them does not matter and so we implemented the pair compression
for all pairs in $\Sigma_\ell\Sigma_r$. The cost of representation and credit of new letters is covered by the released credit,
see Lemma~\ref{lem:noncrossing compression}.
\end{proof}

\begin{corollary}
	\label{cor: pair compression cost}
The compression of pairs issues in total $\Ocomp(m\log N)$ credit during the run of \algmain;
the credit of the new letters introduced to $G$ and their representation costs
are covered by the credit issued or released during \algpaircr.
\end{corollary}

\subsection{Blocks compression}
\label{subsec:block compression}
Similar notions and analysis as the ones for pairs are applied for blocks.
Consider occurrences of maximal $a$-blocks in \mytext{} and their derivation by $G$.
Then a block $a^\ell$ has a \emph{crossing occurrence} in $X_i$ with a rule $X_i \to \alpha_i$,
if it is contained in $\eval(X_i)$ but this occurrence is not
generated by the explicit $a$s in the rule nor in the substrings generated by the nonterminals in $\alpha_i$.
If $a$s blocks have no crossing occurrences, then $a$ \emph{has no crossing blocks}.
As for noncrossing pairs, the compression of $a$ blocks, when it has no crossing blocks, is easy:
it is enough to replace each explicit maximal $a$-block in the rules of $G$.
We use similar terminology as in the case of pairs: we say that a subprocedure properly implements a block compression for $a$.

\begin{algorithm}[H] \normalsize
  \caption{$\algblocks(a)$, which compresses $a$ blocks when $a$ has no crossing blocks \label{ac code}}
  \begin{algorithmic}[1]
  \For{each $a^{\ell_m}$} 
		\State \label{replace non-extendible}
		replace every explicit maximal block $a^{\ell_m}$ in $G$ by $a_{\ell_m}$
  \EndFor
  \end{algorithmic}
\end{algorithm}

\begin{lemma}
If $a$ has no crossing blocks then $\algblocks(a)$ properly implements the $a$'s blocks compression.

Furthermore, if a letter $b$ from $\mytext$ had no crossing blocks in $G$,
it does not have them in $G'$.
\end{lemma}
The proof is similar to the proof of Lemma~\ref{lem:noncrossing compression}
and so it is omitted.
Note that we do not yet discuss the issued credit, nor the cost of the representation of letters
representing blocks (the latter is done in Section~\ref{subsec: blocks representation cost}).

It is left to ensure that no letter has a crossing block.
The solution is similar to \algpop, this time though we need to remove the whole prefix and suffix
from $\eval(X_i)$ instead of a single letter. The idea is as follows: suppose that $a$ has a crossing block
because $aX_i$ occurs in the rule and $\eval(X_i)$ begins with $a$.
Left-popping $a$ does not solve the problem, as it might be that $\eval(X_i)$ still begins with $a$.
Thus, we keep on left-popping until the first letter of $\eval(X_i)$ is not $a$,
i.e.\ we remove the $a$-prefix of $\eval(X_i)$.
The same works for suffixes.
 
\begin{algorithm}[H]
  \caption{$\algremblocks$: removing crossing blocks.
  \label{removing outer letters}}
  \begin{algorithmic}[1]
  \For{$i \gets 1 \twodots m-1$}
		\State let $a$, $b$ be the first and last letter of $\eval(X_i)$ 
		\State let $\ell_i$, $r_i$ be the length of the $a$-prefix and $b$-suffix of $\eval(X_i)$
		\\ \Comment{If $\eval(X_i) \in a^*$ then $r_i = 0$ and $\ell_i = |\eval(X_i)|$}
		\State remove $a^{\ell_i}$ from the beginning and $b^{r_i}$ from the end of $\alpha_i$
		\State replace $X_i$ by $a^{\ell_i}X_ib^{r_i}$ in the rules \label{popping prefix}
		\If{$\eval(X_i) = \epsilon$}
			\State remove $X_i$ from the rules
		\EndIf
	\EndFor
	\end{algorithmic}
\end{algorithm}

\begin{lemma}
\label{lem: no crossing blocks}
After \algremblocks{} no letter has a crossing block and $\eval(X_m) = \eval(X_m')$.
\end{lemma}
\begin{proof}
Firstly, $\eval(X_m') = \eval(X_m)$: observe that when we remove $a$-prefix $a^{\ell_i}$ from $\alpha_i$ we replace each $X_i$ with $a^{\ell_i}X_i$
(ans similarly for the $b$-suffix), also when we remove $X_i$ from the rules then $\eval(X_i) = \epsilon$.
Hence when processing $X_i$, the strings generated by all other nonterminals are not affected.
In particular, as we do not remove the prefix and suffix of $X_m$, the string generated by $X_m$ remains the same after \algremblocks.

By above observation, the value of $\eval(X_i)$ does not change until \algremblocks{} considers $X_i$.
We show that when $\algremblocks$ considers $X_i$ such that $\eval(X_i)$ has $a$-prefix $a^{\ell_i}$ and $b$-suffix $b^{r_i}$,
then $\alpha_i$ begins with $a^{\ell_i}$ and ends with $b^{r_i}$ (the trivial case, when $\eval(X_i) = a^{\ell_i}$ is shown in the same way).
Suppose that this is not the case and consider $X_i$ with smallest $i$ for which this is not true.
Clearly it is not $X_1$, as there are no nonterminals in $\alpha_1$ and so $\eval(X_1) = \alpha_1$.
So let $X_i$ have a rule $X_i \to \alpha_i$, we deal only with the $a$-prefix, the proof of $b$-suffix is symmetrical.
Since the $a$-prefix of $\eval(X_i)$ and $\alpha_i$ are different,
this means that the $a$-prefix of $\eval(X_i)$ is partially generated by the first nonterminal in $\alpha_i$, let it be $X_j$.
By the choice of $i$ we know that $X_j$ popped its prefix (of some letter, say $a'$) and so it was replaced with $a'^{\ell_j}X_j'$.
Furthermore, $\eval(X_j')$ begins with $a'' \neq a'$.
Since there is no way to remove this $a'$ prefix from $\alpha_i$,
this $a'^{\ell_j}$ is part of the $a$-prefix of $\eval(X_i)$, in particular $a' = a$.
However, $\eval(X_j')$ begins with $a'' \neq a$, so the $a$-prefix of $\alpha_i$ and $\eval(X_i)$ is the same, contradiction.

As a consequence, if $aX_i$ occurs in any rule, then $a$ is not the first letter of $\eval(X_i)$,
as prefix of letters $a$ was removed from $X_i$. Other cases are handled similarly.
So there are no crossing blocks after \algremblocks.
\qedhere 
\end{proof}

So the compression of all blocks of letters is done by first running \algremblocks{}
and then compressing each of the block by \algblocks.
Note that we do not compress blocks of letters that are introduced in this way.
Concerning the number of credit, the arbitrary long blocks popped by \algremblocks{}  
are compressed (each into a single letter) and so at most $8$ credit per rule is issued.

\begin{algorithm}[H] \normalsize
  \caption{\algblocksc: compresses blocks of letters}
  \begin{algorithmic}[1]
	\State run \algremblocks
	\State $L \gets$ list of letters in \mytext{}
  \For{each $a \in L$} 
		\State run $\algblocks(a)$
  \EndFor
  \end{algorithmic}
\end{algorithm}

\begin{lemma}
\label{lem:blocksc}
\algblocksc{} properly implements the blocks compression for each letter $a$ occurring in \mytext{} before its application
and issues $\Ocomp(m)$ credit.
The issued credit covers the cost of credit of letters introduced during the \algblocksc{}
(but not their representation cost).
\end{lemma}
The proof is similar as the proof of Lemma~\ref{lem:pc crossing} so it is omitted.

\begin{corollary}
\label{cor: block compression cost}
During the whole \algmain{} the \algblocksc{} issues in total at most $\Ocomp(m \log N)$ credit.
The credit of the new letters introduced to $G$ is covered by the issued credit.
\end{corollary}
Note that the cost of representation of letters replacing blocks is not covered by the credit,
this cost is separately estimated in the next subsection.

\subsection{Calculating the cost of representing letters in block compression}
\label{subsec: blocks representation cost}
The issued credit is enough to pay the 2 credit for occurrences of letters introduced during \algmain{}
and the released credit is enough to pay the credit of the letters introduced during the pair compression and their representation cost.
However, credit alone cannot cover the representation cost of letters replacing blocks.
The appropriate analysis is presented in this section. The overall plan is as follows:
firstly, we define a scheme of representing the letters based on the grammar $G$
and the way $G$ is changed by \algblocksc{} (the \emph{$G$-based representation}).
Then for such a representation schema, we show that the cost of representation is $\Ocomp(g \log N)$.
Lastly, it is proved that the actual cost of representing the letters by \algmain{} (the \emph{\algmain-based representation})
is smaller than the $G$-based one, hence it is also $\Ocomp(g \log N)$.

\subsubsection{$G$-based representation}
The intuition is as follows: while the $a$ blocks can have exponential length,
most of them do not differ much, as in most cases the new blocks are obtained by
concatenating letters $a$ that occur explicitly in the grammar
and in such a case the released credit can be used to pay for the representation cost.
This does not apply when the new block is obtained by concatenating two different blocks of $a$
(popped from nonterminals) inside a rule.
However, this cannot happen too often: when blocks of length $p_1$, $p_2$, \ldots, $p_\ell$
are compressed (at the cost of $\Ocomp\Big(\sum_{i=1}^\ell \left(1 + \log p_i\right)\Big) = \Ocomp(\log (\prod_{i=1}^\ell p_i))$,
as each $p_i \geq 2$), the length of the corresponding text in the input text is $\prod_{i=1}^\ell p_i$, which is at most $N$.
Thus $\Ocomp\Big(\sum_{i=1}^\ell (1 + \log p_i) \Big)= \Ocomp(\log \prod_{i=1}^\ell p_i) = \Ocomp(\log N$) cost per nonterminal is scored.

We create a new letter for each $a$ block in the rule $X_i \to \alpha_i$
after \algremblocks{} popped prefixes and suffixes from $X_1,\ldots, X_{i-1}$ but before it popped letters from $X_i$.
(We add the artificial empty block $\epsilon$ to streamline the later description and analysis.)
Such a block is a \emph{power} if
it is obtained by concatenation of two $a$-blocks popped from nonterminals inside a rule (and perhaps some other explicit letters $a$),
note that this power may be then popped from a rule (since it is a prefix or suffix in this rule).
This implies that in the rule $X_i \to u X_j v X_k w$ the popped suffix of $X_j$ and
popped prefix of $X_k$ are blocks of the same letter, say $a$, and furthermore $v \in a^*$.
Note that it might be that one (or both) of $X_j$ and $X_k$ were removed in the process
(in this case the power can be popped from a rule as well).
For each block $a^{\ell}$ that is not a power we may uniquely identify another block $a^k$
(perhaps $\epsilon$, not necessarily a power)
such that $a^\ell$ was obtained by concatenating $\ell - k$ explicit letters to $a^k$ in some rule.
\begin{lemma}
	\label{lem: not powers}
For each block $a^{\ell}$ represented in the $G$-based representation that is not a power there is block $a^k$ 
(perhaps $k = 0$) such that $a^k$ is also represented in $G$-based representation
and $a^\ell$ was obtained in a rule by concatenating $\ell - k$ explicit letters that existed in the rule to $a^k$.
\end{lemma}
Note that the block $a^k$ is not necessarily unique: it might be that there are several $a^\ell$ blocks in $G$
which are obtained as different concatenations of $a^k$ and $\ell - k$ explicit letters.
\begin{proof}
Let $a_\ell$ be created in the rule for $X_i$, after popping prefixes and suffixes from $X_1,\ldots,X_{i-1}$.
Consider, how many popped prefixes and suffixes take part in this $a^\ell$.

If two, then it is a power, contradiction.

If one, then let the popped prefix (or suffix) be $a^k$.
Since it was popped, say from $X_j$, then $a^k$ was a maximal block in $X_j$ before popping, so it is represented as well.
Then in the rule for $X_i$ the $a^\ell$ is obtained by concatenating $\ell - k$ letters $a$ to $a^k$.
None of those letters come from popped prefixes and suffixes, so they are all explicit letters that were present in this rule.

If there are none popped prefixes and suffixes that are part of this $a^\ell$,
then all its letters are explicit letters from the rule for $X_i$,
and we treat it as a concatenation of $k$ explicit letters to $\epsilon$.
\qedhere
\end{proof}

We represent the blocks as follows:
\begin{enumerate}
	\item \label{power}for a block $a^\ell$ that is a power we represent $a_\ell$ using the binary expansion,
	which costs $\Ocomp(1 + \log \ell)$;
	\item \label{credit}for a block $a^\ell$ that is obtained by concatenating $\ell - k$ explicit letters
	to a block $a^k$ (see Lemma~\ref{lem: not powers})
	we represent $a_\ell$ as $a_k a\ldots a$ which has a representation cost of $\ell - k + 1$,
	this cost is covered by the $2(\ell-k) \geq \ell - k + 1$ credit released by the $\ell - k$ explicit letters $a$.
	Note that the credit released by those letters was not used for any other purpose.
	(Furthermore recall that the $2$ units of credit per occurrence of $a_\ell$ in the rules of grammar are already covered 
	by the credit issued by \algblocksc, see Lemma~\ref{lem:blocksc}.)
\end{enumerate}
We refer to cost in \ref{power} as the \emph{cost of representing powers}
and redirect this cost to the nonterminal in whose rule this power is created.
The cost in \ref{credit}, as marked there, is covered by released credit.

\subsubsection{Cost of $G$-based representation}
We now estimate the cost of representing powers.
The idea is that if nonterminal $X_i$ is charged the cost of representing powers of length
$p_1$, $p_2$, \ldots, $p_\ell$, which have representation cost $\Ocomp(\sum_{i=1}^\ell \log p_i) = \Ocomp(\log (\prod_{i=1}^\ell p_i))$,
then in the input this nonterminal generated a text of length at least $p_1 \cdot p_2 \cdots p_\ell \leq N$
and so the total cost of representing powers is $\Ocomp(\log N)$ (per nonterminal).
This is formalised in the lemma below.

\begin{lemma}
\label{clm:cost from power to rule}
The total cost of representing powers by $G$-based representation charged towards a single rule is $\Ocomp(\log N)$.
\end{lemma}
\begin{proof}
There are two cases:
first, after the creation of the power in a rule $X_i \to u X_j v X_k w$ one of the nonterminals
$X_j$, $X_k$ is removed. But this happens at most once for the rule and the cost of $\Ocomp(\log N)$
of representing the power can be charged to a rule.

The second and crucial case is when after the creation of power
both nonterminals remained in a rule $X_i \to u X_j v X_k w$.
Note that creation of the $a$ power here means that
$\eval(X_j)$ has $a$-suffix, $\eval(X_k)$ an $a$-prefix and  $v \in a^*$.

Fix this rule and consider all such creations of powers performed on this rule.
Let the consecutive letters, whose blocks are compressed, be
$a^{(1)}$, $a^{(2)}$, \ldots, $a^{(\ell)}$ and their lengths $p_1$, $p_2$, \ldots, $p_\ell$.
Lastly, the $p_\ell$ repetitions of $a^{(\ell)}$ are replaced by $a^{(\ell+1)}$.
(Observe, that $a^{(i+1)}$ does not need to be the letter that replaced the $a^{(i)}$'s block,
as there might have been some other compression performed on that letter.)
Then the cost of the representing powers is constant time more than
\begin{equation}
	\label{eq:representation cost blocks}
	\sum_{i=1}^\ell (1 + \log p_i) \leq 2 \sum_{i=1}^\ell \log p_i \enspace .
\end{equation}

Define \emph{weight}: for a letter it is the length of the
substring of the original input string that it `derives'.
Note that the maximal weight of any letter is $N$, the length of the input word.

Consider the weight of the strings between $X_j$ and $X_k$.
Clearly, after the $i$-th blocks compression it is exactly $p_i \cdot \weight(a^{(i)})$,
as the block of $p_i$ letters $a^{(i)}$ was replaced by one letter.
We claim that $\weight(a^{(i+1)}) \geq p_i \weight(a^{(i)})$:
right after the $i$-th blocks compression the string between $X_j$ and $X_k$
is simply a letter $a^{(i)}_{p_i}$, which replaced the $p_i$ block of $a^{(i)}$.
After some operations, this string consists of $p_{i+1}$ letters $a^{(i+1)}$.
Observe that $(a^{(i+1)})^{p_{i+1}}$ `derives' $a^{(i)}_{p_i}$: indeed all operations performed by \algmain{} do not remove
the letters from string between $X_j$ and $X_k$ in a rule, only replace strings with single letters
and perhaps add letters at the ends of this string.
But if $(a^{(i+1)})^{p_{i+1}}$ `derives' $a^{(i)}_{p_i}$, i.e.\ a single letter, then also $a^{(i+1)}$ `derives' $a^{(i)}_{p_i}$,
hence 
$$
\weight(a^{(i+1)}) \geq \weight(a^{(i)}_{p_i}) = p_i \weight(a^{(i)}) \enspace .
$$
Since $\weight(a^{(1)}) \geq 1$ it follows that $\weight(a^{(\ell+1)}) \geq \prod_{i=1}^\ell p_i$.
As $\weight(a^{(\ell+1)}) \leq N$ we have
\begin{align*}
N &\geq \prod_{i=1}^\ell p_i
\intertext{and so it can be concluded that}
\log(N) &\geq \log \left( \prod_{i=1}^\ell p_i \right)\\
		&= \sum_{i=1}^\ell \log p_i \enspace .
\end{align*}
Therefore, the whole cost $\sum_{i=1}^\ell \log p_i$, as estimated in~\eqref{eq:representation cost blocks},
is $\Ocomp(\log N)$, as claimed.
\qedhere
\end{proof}

\begin{corollary}
\label{cor: G representation cost}
The cost of $G$-based representation is $\Ocomp(g + g \log N)$.
\end{corollary}
\begin{proof}
Concerning the cost of representing powers, by Lemma~\ref{clm:cost from power to rule}
we redirect at most $\Ocomp(\log N)$ against each of the $m \leq g$ rules of $G$.
The cost of representing non-powers is covered by the released credit;
the initial value of credit is at most $2g$ and by Corollary~\ref{cor: pair compression cost}
and Corollary~\ref{cor: block compression cost} at most $\Ocomp(g \log N)$ credit is issued during the whole run of \algmain,
which ends the proof.
\qedhere
\end{proof}

\subsubsection{Comparing the $G$-based representation cost  and \algmain-based representation cost}
We now show that the cost of \algmain-based representation
is at most as high as $G$-based one.
We first represent $G$-based representation cost using a weighted graph $\mathcal G_G$,
such that the $G$-based representation is (up to a constant factor) $w(\mathcal G_G)$,
i.e.\ the sum of weights of edges of $\mathcal G_G$.

\begin{lemma}
	\label{lem: graph representation}
The cost of $G$-based representation of all blocks is $\Theta(w(\mathcal G_G))$,
where nodes of $\mathcal G_G$ are labelled with blocks represented in the $G$-based representation
and edge from $a^\ell$ to $a^k$, where $\ell > k$, has weight $\ell - k$ or $1 + \log(\ell - k)$ (in this case additionally $k = 0$).
Each node has at least one outgoing edge.

The former corresponds to the representation cost covered by the released credit while the latter to the cost of representing powers.
\end{lemma}
\begin{proof}
We give a construction of the graph $\mathcal G_G$.

Fix the letter $a$ and consider any of the blocks $a^\ell$ that is represented by $G$,
we put a node $a^\ell$ in $\mathcal G_G$.
Note that a single $a^\ell$ may be represented in many ways: different occurrences of $a^\ell$
are replaced with $a_\ell$ and may be represented in different ways (or even twice in the same way),
this means that $\mathcal G_G$ may have more than one outgoing edge per node. 
\begin{itemize}
	\item when $a^\ell$ is a power, we create an edge from the node labelled with $a^\ell$ to $\epsilon$,
 the weight is $1 + \log \ell$ (recall that this is the cost of representing this power);
	\item when $a_\ell$ is represented as a concatenation of $\ell - k$ letters to $a_k$,
	we create and edge from the node $a^\ell$ to $a^k$, the weight is $\ell - k$ (this is the cost of representing this block;
	it was paid by the credit on the $\ell - k$ explicit letters $a$).
\end{itemize}
Then the sum of the weight of the created graph is a cost of representing the blocks using the $G$-based representation
(up to a constant factor).
\qedhere
\end{proof}

Similarly, the cost of $\algmain$-based representation has a graph representation $\mathcal G_\algmain$.

\begin{lemma}
\label{lem: algmain representation cost}
The cost of \algmain-representation for blocks of a letter $a$ is $\Theta(w(\mathcal G_\algmain))$,
where the nodes of $\mathcal G_\algmain$ are labelled with blocks represented by \algmain-representation and it has an edge from $a^\ell$ to $a^k$
if and only if $\ell$ and $k$ are two consecutive lengths of $a$-blocks.
Such an edge has weight $1 + \log (\ell - k)$.
\end{lemma}
\begin{proof}
Observe that this is a straightforward consequence of the way the blocks are represented:
Lemma~\ref{lem: cost of powers} guarantees that when blocks $a^{\ell_1}, a^{\ell_2}, \ldots, a^{\ell_k}$
(where $1 < \ell_1 < \ell_2< \dots < \ell_k$)
are represented the \algmain-representation cost is $ \Ocomp(\sum_{i=1}^k[1 + \log(\ell_{i} - \ell_{i-1})])$,
so we can assign cost $1 + \log (\ell_{i} - \ell_{i-1})$ to $a^{\ell_i}$ (and make it the weight on the edge to the previous block).
\qedhere
\end{proof}

We now show that $\mathcal G_G$ can be transformed to $\mathcal G_\algmain$ without increasing the sum of weights of the edges.

\begin{lemma}
	\label{lem: transforming representations}
$\mathcal G_G$ can be transformed to $\mathcal G_\algmain$ without increasing the sum of weights of the edges.
\end{lemma}
\begin{proof}
Fix a letter $a$, we show how to transform the subgraph of $\mathcal G_G$ induced by nodes labelled with blocks of $a$
to the corresponding subgraph of $\mathcal G_\algmain$, without increasing the sum of weights.

Firstly, let us sort the nodes according to the increasing length of the blocks.
For each node $a^\ell$, if it has many edges, we delete all except one and then we redirect this edge to $a^\ell$'s direct predecessor
(say $a^k$) and label it with a cost $1 + \log (\ell-k)$. This cannot increase the sum of weights of edges:
\begin{itemize}
	\item deleting does not increase the sum of weights;
	\item if $a_\ell$ has an edge to $\epsilon$ with weight $1 + \log \ell$ then $1 + \log \ell \geq 1 + \log (\ell - k)$;
	\item otherwise it had an edge to some $k' \leq k$ with a weight $\ell - k'$. Then
	$1 + \log (\ell - k) \leq \ell - k \leq \ell - k'$, as claimed (note that $1 + \log x \leq x$ for $x \geq 1$). 
\end{itemize}
Some blocks labelling nodes in $\mathcal G_G$ perhaps do not label the nodes in $\mathcal G_\algmain$.
For such a block $a^\ell$ we remove its node $a_\ell$ and redirect its unique incoming edge to its predecessor,
say $a_{\ell'}$, changing the weight appropriately.
Since $1 + \log (x) + 1 + \log(y) > 1 + \log(x+y)$ when $x, y \geq 1$, we do not increase the total weight.

It is left to observe that if a node labelled with $a^\ell$ exists in $\mathcal G_\algmain$ then it also exists
in $\mathcal G_G$, i.e.\ 
all blocks represented in \algmain{} occur in \mytext.
After \algremblocks{} there are no crossing blocks, see Lemma~\ref{lem: no crossing blocks}.
So any maximal block in \mytext{} (i.e.\ one represented by \algmain-based representation)
is also a maximal block $a^\ell$ in some rule (after \algremblocks), say in $X_i$.
But then this block is present in $X_i$ also just before action of \algremblocks{} on $X_i$ and so it is represented by $G$-based representation.

In this way we obtained a graph corresponding to the \algmain-based representation.
\end{proof}

\begin{corollary}
\label{cor: algmain representation cost}
The total cost of \algmain-representation is $\Ocomp(g \log N)$.
\end{corollary}
\begin{proof}
By Lemma~\ref{lem: transforming representations} it is enough to show this for the $G$-based representation,
which holds by Corollary~\ref{cor: G representation cost} 
\qedhere
\end{proof}

\section{Improved algorithm and its analysis}
The naive algorithm, which simply represents the input word $w$ as $X_1 \to w$ results in a grammar of size $N$.
In some extreme cases this might be better than $\Ocomp(g \log N)$ guaranteed by \algmain.
We merge the naive approach with the recompression-based algorithm:
if at the beginning of a phase $i$ \algmain{} already paid $k_i$ for representation of the letters
and the remaining text is $\mytext_i$
then we can construct a grammar for the input string of the total size $k_i + |\mytext_i|$ by giving a rule $X \to T_i$.
Of course we can then choose the minimum over all possible $i$
(observe that for $i = 0$  this is simply the naive representation $X \to w$ and for the last $i$
this is the grammar returned by \algmain).
We call the corresponding algorithm \algmaini.
Additionally, we show that when $|\mytext_i| \approx g$
then the so-far cost of representing letters is $\Ocomp(g \log (N / g))$
and so the corresponding grammar considered by \algmaini{} is of size $\Ocomp(g + g \log(N/g))$,
consequently, the grammar returned by \algmaini{} is also of this size.
This matches the best known results for the smallest grammar problem~\cite{SLPaprox,SLPaprox2,SLPaproxSakamoto}.

\begin{algorithm}[H]
	\caption{\algmaini: improved version outline}
	\label{alg:main improved}
	\begin{algorithmic}[1]
	\State $i \gets 0$
	\While{$|\mytext|>1$} 
		\State $\mysize[i] \gets |\mytext| + $ so-far cost of representing letters \Comment{Cost of grammar in phase $i$}
		\State $i \gets i +1$ \Comment{Number of the phase}
		\State $L \gets $ list of letters in \mytext{} \Comment{The compression is done as in \algmain}
    \For{each $a \in L$}
    		\State compress maximal blocks of $a$ 
    \EndFor				    
		\State $P \gets $ list of pairs 
		\State find partition of $\Sigma$ into $\Sigma_\ell$ and $\Sigma_r$ 
    \For{$ab\in P \cap \Sigma_\ell\Sigma_r$}
    	\State compress pair $ab$ 
    \EndFor
	\EndWhile
	\State output grammar $G_i$ for which $\mysize[i]$ is smallest
	 \end{algorithmic}
\end{algorithm}

The properties of \algmaini{} are summarised in the following theorem.

\begin{theorem}
	\label{thm: main imp}
The \algmain{} runs in linear time and returns a grammar of size $\Ocomp\left(g + g \log\left(\frac{N}{g} \right)\right)$,
where $g$ is the size of the optimal grammar for the input text.
\end{theorem}
The time analysis follows in the same way as in the case of \algmain{} (the only additional computation is storing the sizes and choosing the minimum of them),
so it is omitted. In the rest of this section we show the bound on the size of the returned grammar.

In the following analysis we focus on the phase $i$ such that $T_{i} \geq g > T_{i+1}$
(for input text with more than one symbol such an $i$ exists, as for the `last' $i$ we have $T_i = 1$).
Then we separately estimate the cost of representation
(i.e.\ issued credit and the cost of $\algmain$-based representation)
up to phase $i$ and in the phase $i+1$.
We show that both of those are $\Ocomp(g + g \log (N/g))$, which shows the main claim of Theorem~\ref{thm: main imp}.

\begin{lemma}
\label{lem:cost of compression}
If at the beginning of the phase $|\mytext| \geq g$ then so far the cost of representing letters by \algmaini{}
as well as the credit on $G$ is $\Ocomp(g + g \log(N/g))$.
\end{lemma}
\begin{proof}

We estimate separately the amount of issued credit and the cost of representation of letters replacing blocks.
This covers the whole cost of representing letters
(see Corollary~\ref{cor: pair compression cost}, Corollary~\ref{cor: block compression cost})
as well as the credit on the letters in the grammar.

\subsubsection*{Credit}
Observe first that initial grammar $G$ has at most $g$ credit.
The input text is of length $N$ and the current one is of $t = |\mytext|$
and so there were $\Ocomp(\log(N/t))$ phases, as in each phase the length of $\mytext$
drops by a constant factor, see Lemma~\ref{lem:number of phases}.
As $t \geq g$, we obtain a bound $\Ocomp(\log(N/g))$ on the number of phases.
Due to Lemmata~\ref{lem:pc crossing}, \ref{lem:blocksc}, at most $\Ocomp(m)$ credit per phase is issued
during the pair compression and block compression,
so in total $\Ocomp(g + g \log(N/g))$ credit was issued.
From Corollary~\ref{cor: pair compression cost} and Corollary~\ref{cor: block compression cost}
we conclude that this credit is enough to cover
the credit of all letters as well as the representation cost of letters introduced during the pair compression.
So it is left to calculate the cost of representing blocks.

\subsubsection*{Representing blocks}
The representation of blocks used by \algmaini{} is the same as the one of \algmain.
So we can define the $G$-based representation in the same way as previously.
For both the $G$-based representation and the $\algmaini$-based representation
we can define graphs $\mathcal G_G$ and $\mathcal G_\algmaini$
and by Lemma~\ref{lem: graph representation} the cost of $G$-based representation is $\Theta(w(\mathcal G_G))$
and by Lemma~\ref{lem: algmain representation cost} the cost of \algmaini-based representation is $\Theta(w(\mathcal G_\algmaini))$.
Then Lemma~\ref{lem: transforming representations} shows that we can transform $\mathcal G_G$ to $\mathcal G_\algmaini$ without increasing
the sum of weights.
Hence it is enough to show that the $G$-based representation cost is at most $\Ocomp(g + \log (N/g))$.

The $G$-based representation cost consists of some released credit and the cost of representing powers,
see Lemma~\ref{lem: graph representation}.
The former was already addressed (the whole issued credit is $\Ocomp(g + g \log(N/g))$)
and so it is enough to estimate the latter, i.e.\ the cost of representing powers.

The outline of the analysis is as follows:
when a new power $a^\ell$ is represented,
we mark some letters of the input text (and perhaps modify some other markings)
those markings are associated with nonterminals
and are named $X_i$-pre-power marking and $X_i$-in marking
(which are defined in more detail later on).
The markings satisfy the following conditions:
\begin{enumerate}[(M1)]
	\item \label{M1} each marking marks at least $2$ letters, no two markings mark the same letter;
	\item \label{M2} for each $X_i$ there is most one $X_i$-pre-power marking and at most one $X_i$-in marking;
	\item \label{M3} 
	when the substrings of length $p_1$, $p_2$, \ldots, $p_k$ are marked,
then the so-far cost of representing the powers by $G$-based representation is $c \sum_{i = 1}^{k} (1 + \log p_i)$
(for some fixed constant c).
\end{enumerate}

Using \Mrefall{} the total cost of representing powers (in $G$-based representation) can be upper-bounded by
(a constant times):
\begin{subequations}
\label{eq:estimations}
\begin{equation}
\label{eq:estimations 1}
k + \sum_{i = 1}^{k} \log p_i, \text{ where } k \leq 2m \text{ and } \sum_{i = 1}^{k} p_i \leq N \enspace .
\end{equation}
It is easy to show that~\eqref{eq:estimations 1} is maximised for $k = 2m$ and each $p_i$ equal to $N / 2m$:
clearly, the sum is maximised for $\sum_{i = 1}^{k} p_i = N$.
Then for a fixed $k$ and $\sum_{i = 1}^{k} p_i = N$
the sum $\sum_{i = 1}^{k} \log p_i$ is maximised when all $p_i$ are equal, which follows from
the fact that $\log(x)$ is concave, hence we can set $p_i = \frac{N}{k}$.
Lastly, the $k + k \log (N / k)$ has a non-negative derivative (for $k$) and so (weakly) increases
with $k$. Since $k \leq 2m \leq 2g$, this is maximised for $k = 2g$.
In this way the value of~\eqref{eq:estimations 1} is at most
\begin{equation}
\label{eq:estimations final}
2g + 2g \log \left(\frac{N}{2g}\right) = \Ocomp\left(g + g \log \left(\frac{N}{g} \right)\right) \enspace .
\end{equation}
\end{subequations}

The idea of preserving~\Mrefall{} is as follows: if a new power of length $\ell$ is represented,
this yields a~cost $\Ocomp(1 + \log \ell) = \Ocomp(\log \ell)$, see Lemma~\ref{lem: graph representation};
we can choose $c$ in \Mref{3} so that this is at most $c \log \ell$ (as $\ell \geq 2)$.
Then either we mark new $\ell$ letters or we remove some marking of length $\ell'$ and mark $\ell \cdot \ell'$ letters,
it is easy to see that in this way~\Mrefall{} is preserved (still, those details are repeated later in the proof).

Whenever we are to represent powers $a^{\ell_1}$, $a^{\ell_2}$, \ldots, for each power $a^\ell$,
where $\ell > 1$, we find the right-most maximal block $a^\ell$ in \mytext.
It is possible that this particular $a^\ell$ was obtained as a concatenation of $\ell - k$ explicit letters to $a^k$
(so, not as a power). In such a case we are lucky, as the representation of this $a_\ell$ is paid by the credit
and we do not need to separately consider the cost of representing power $a^\ell$.
Otherwise the $a^\ell$ in this rule is obtained as a power and we mark some of the letters in the input that are `derived' by this $a^\ell$.
The type of marking depends on the way this particular $a^\ell$ is `derived':
Let $X_i$ be the smallest nonterminal that derives (before \algremblocks)
this right-most occurrence of maximal $a^\ell$ (clearly there is such non-terminal, as $X_m$ derives it).
If one of the nonterminals in $X_i$'s production was removed during \algremblocks,
this marking is an \emph{$X_i$-pre-power marking}.
Otherwise, this marking is an \emph{$X_i$-in marking}.

\begin{clm}
\label{clm: markings}
There is at most one $X_i$-pre-power marking.

When $X_i$-in marking is created for $a^\ell$, after the block compression $X_i$ has two nonterminals inside its rule
and between them there is exactly $a_\ell$.
\end{clm}

\begin{proof}
Concerning the $X_i$-pre-power marking, let $a^\ell$ be the first power that gets this marking.
Then be definition of the marking, afterwards in the rule for $X_i$ there is only one nonterminal.
But this means that no power can be created in this rule later on,
in particular, no new marking associated with $X_i$ (pre-power marking or in marking) can be created.

Suppose that $a^\ell$ was assigned an $X_i$-in marking, which as in the previous case means
that the right-most occurrence of maximal block $a^\ell$ is generated by $X_i$ but not by the
nonterminals in the rule for $X_i$.
Since $a^\ell$ is a power it is obtained in the rule as a concatenation of the $a$-prefix and the $a$-suffix popped from nonterminals in the rule for $X_i$.
In particular this means that each nonterminal in the rule for $X_i$ generate a part of this right-most occurrence of $a^\ell$.
If any of those nonterminals were removed during the block compression $a^\ell$ would be assigned an $X_i$-pre-power marking,
which is not the case.
So both those nonterminals remained in the rule.
Hence after popping prefixes and suffixes, between those two nonterminals there is exactly a block $a^\ell$, which is then
replaced by $a_\ell$, as promised, which ends the proof.
\qedhere
\end{proof}

Consider the $a^\ell$ and the `derived' substring $w^\ell$ of the \emph{input text}.
We show that if there are markings inside $w^\ell$, they are all inside the last among those $w$s.

\begin{clm}
\label{clm: existing marking}
Let $a^\ell$ be an occurrence of a maximal block to be replaced with $a_\ell$ which `generates' $w^\ell$ in the input text.
If there is any marking within this $w^\ell$ then it is within the last among those $w$s.
\end{clm}
\begin{proof}
Consider any pre-existing marking within $w^\ell$, say it was done when some $b^k$ was replaced by $b_k$.
As $b_k$ is a single letter and $a^\ell$ derives it, each $a$ derives at least one $b_k$.
The marking was done inside the string generated by the right-most $b_k$
(as we always put the marking within the rightmost occurrence of the string to be replaced).
Clearly the right-most $b_k$ is `derived' by the right-most $a$ within $a^\ell$,
sin in particular it is inside the right-most $w$ in this $w^\ell$.
So all markings within $w^\ell$ are in fact within the right-most $w$.
\qedhere
\end{proof}

We now demonstrate how to mark letters in the input text.
Suppose that we replace a power $a^\ell$,
let us consider the right-most occurrence of this $a^\ell$ in $\mytext$
and the smallest $X_i$ that generates this occurrence.
This $a^\ell$ generates some $w^\ell$ in the input text.
If there are no markings inside $w^\ell$ then we simply mark any $\ell$ letters within $w^\ell$.
In the other case, by Claim~\ref{clm: existing marking} we know that all those markings are in fact in the last $w$.
If any of them is the (unique) $X_i$-in marking, let us choose it. Otherwise choose any other marking.
Let $\ell'$ denote the length of the chosen marking.
Consider, whether this marking in $w$ is unique or not

\begin{description}
	\item[unique marking] Then we remove it and mark arbitrary $\ell \cdot \ell'$ letters in $w^\ell$;
this is possible, as $|w| \geq \ell'$ and so $|w^\ell| \geq \ell \cdot \ell'$.
Since $\log(\ell \cdot \ell') = \log \ell + \log \ell'$, the~\Mref{3} is preserved,
as it is enough to account for the $1 + \log \ell \leq c \log  \ell$ representation cost of $a^\ell$
as well as the $c \log \ell'$ cost associated with the previous marking of length $\ell'$.
	\item[not unique] Then $|w| \geq \ell' + 2$ (the $2$ for the other markings, see~\Mref{1}).
We remove the marking of length $\ell'$, let us calculate how many unmarked letters are in $w^\ell$ afterwards:
in $w^{\ell-1}$ there are at least $(\ell-1) \cdot (\ell'+2)$ letters (by the Claim~\ref{clm: existing marking}: none of them marked)
and in the last $w$ there are at least $\ell'$ unmarked letters (from the marking that we removed):
\begin{align*}
(\ell-1) \cdot (\ell'+2) + \ell'
	&=
(\ell \ell' + 2 \ell - \ell' - 2) + \ell'\\
	&=
\ell \ell' + 2 \ell - 2\\
	&>
\ell \ell'\enspace.
\end{align*}
We mark those $\ell \cdot \ell'$ letters, as in the previous case, the associated $c \log( \ell \ell')$
is enough to pay for the cost.
\end{description}

There is one issue: it might be that we created an $X_i$-in marking while there already was one, violating~\Mref{2}.
However, we show that if there were such a marking, it was within $w^\ell$ (and so within the last $w$, by Claim~\ref{clm: existing marking})
and so we could choose it as the marking that was deleted when the new one was created.
Consider the previous $X_i$-in marking.
It was introduced for some power $b^k$, replaced by $b_{k}$ that was a unique letter between the
nonterminals in the rule for $X_i$, by Claim~\ref{clm: markings}.
Consider the rightmost substring of the input text that is generated by the explicit letters
between nonterminals in the rule for $X_i$.
The operations performed on $G$ cannot shorten this substring, in fact they often expand it.
When $b_k$ is created, this substring is generated by $b_k$, by Claim~\ref{clm: markings}.
When $a_\ell$ is created, it is generated by $a_\ell$, by Claim~\ref{clm: markings}, i.e.\ this is exactly $w^\ell$.
So in particular $w^\ell$ includes the marking for $b_k$.

This shows that \Mrefall{} holds and so also the calculations in~\eqref{eq:estimations} hold,
in particular, the representation cos of powers is $\Ocomp(g \log (N/g))$.  
\end{proof}

Let $t_1$ and $t_2$ be the lengths of $|\mytext|$ at the beginning of two consecutive
phases, such that  $t_1 \geq g > t_2$.
By Lemma~\ref{lem:cost of compression} the cost of representing letters and the credit before the $|\mytext|$
was reduced to $t_2$ letters (as well as the credit remaining on the letters of grammar) is $\Ocomp(g + g \log(N/g))$.
So it is left to estimate what is the cost of representation in this phase.

\begin{lemma}
\label{lem:last phase estimations}
Consider a phase, such that at its beginning \mytext{} has length $t_1$
and after it it has length $t_2$, where $t_1 \geq g > t_2$.
Then the cost of representing letters introduced during this phase is at most $\Ocomp(g + g\log(N/g))$.
\end{lemma}
\begin{proof}
The cost of representing letters introduced during the pair compression is covered by the released credit, see Lemma~\ref{lem:pc crossing}.
There was at most $\Ocomp(g + g\log(N/g))$ credit in the grammar at the beginning of the phase, see Lemma~\ref{lem:cost of compression},
and during this phase at most $\Ocomp(g)$ credit was issued, see Lemma~\ref{lem:pc crossing} and Lemma~\ref{lem:blocksc}.

Consider the cost of representing blocks.
Note that since \mytext{} at the end of the phase has $t_2$ letters, at most
$2t_2$ letters representing blocks could be introduced in this phase
(since at most two blocks can be merged into one letter by pair compression afterwards).
Let $p_1$, \ldots, $p_k$ be the lengths of those powers.
Then the cost of representing them is proportional to (see Lemma~\ref{lem: cost of powers})
$$
k +	\sum_{i=1}^k \log p_i, \text{ where } k\leq 2t_2 \text{ and } \sum_{i=1}^k p_i \leq t_1\enspace .
$$
Since $k \leq 2t_2 < 2g$ we only estimate the sum.
Using the same analysis as in the case of~\eqref{eq:estimations} it can be concluded that this is at most
$$
	2 t_2 \log\left(\frac{t_1} {2t_2}\right) \leq 2 t_2 \log\left(\frac{N}{2t_2}\right) < 2 g \log\left(\frac{N} {2g}\right) 
	=\Ocomp\left(g \log\left(\frac{N}{g}\right)  \right) \enspace ,
$$
with the first equality following from $t_1 \leq N$ and the second from $g > t_2$
and monotonicity of $f(x) = x \log(N/x)$.
\end{proof}

Now the estimations from Lemma~\ref{lem:cost of compression} and 
Lemma~\ref{lem:last phase estimations} allow the proof of Theorem~\ref{thm: main imp}.

\begin{proof}[of Theorem~\ref{thm: main imp}]
The estimation of the running time is the same as in the case of \algmain, so it is omitted.

Concerning the size of the returned grammar,
consider the phase, such that before it the \mytext{} had length $t_1$ and right after it $t_2$,
where $t_1 \geq g > t_2$, there is such a phase as in the end the \mytext{} has length $1$.
Then by Lemma~\ref{lem:cost of compression} the cost of representing letters introduced before this phase is
$\Ocomp\left(g + g \log\left(\frac{N}{g} \right)\right)$ while by Lemma~\ref{lem:last phase estimations}
the cost of representing letters introduced in this phase is at most
$\Ocomp\left(g + g \log\left(\frac{N}{g} \right)\right)$.
Hence the size of the grammar that is calculated by \algmaini{}
after this phase is at most $\Ocomp\left(g + g \log\left(\frac{N}{g} \right)\right)$.
So also the minimum found during the computation is of at most this size.
\end{proof}

\subsubsection*{Acknowledgements}
I would like to thank Pawe\l{} Gawrychowski for introducing me to the topic,
for pointing out the relevant literature~\cite{MehlhornSU97}
and discussions; Markus Lohrey for suggesting the topic of this paper
and bringing the idea of applying the recompression to the smallest grammar.

\appendix

\section{Sakamoto's algorithm~\cite{SLPaproxSakamoto}}
In proof that bounds the number of introduced nonterminals~\cite[Theorem 2]{SLPaproxSakamoto},
it is first estimated that in one execution of the while loop for a factor $f_i$ the introduced nonterminals
occur in $f_1f_2\cdots f_{i-1}$, except perhaps a constant number of them.
This argument follows from observation that $f_i$ is compressed to $\alpha\beta\gamma$,
where $|\alpha|$ and $|\gamma|$ are bounded by a constant and the earlier occurrence of the same string as $f_i$
is compressed to $\alpha'\beta\gamma'$ (where also $|\alpha'|$ and $|\gamma'|$ are bounded by a constant).
This is true, however, when $\alpha$ and $\gamma$ represent nonterminals introduced by \textit{repetition}
procedure (i.e.\ they are blocks in the terminology used here)
we need to take into the account also the additional nonterminals that are introduced for representation of those blocks.
The estimation of $\Ocomp(1)$ is not enough,
as in the worst case $\Omega(\log N)$ are needed to represent a single block of $a$s.
We do not see any easy patch to repair this flaw.

The improved analysis~\cite[Theorem 2]{SLPaproxSakamoto}, in which the number of nonterminals is bounded
by $\Ocomp\left(g + \log\left(\frac{N}{g}\right)\right)$, has the same shortcoming.

\begin{thebibliography}{10}

\bibitem{SLPaprox2}
Moses Charikar, Eric Lehman, Ding Liu, Rina Panigrahy, Manoj Prabhakaran, Amit
  Sahai, and Abhi Shelat.
\newblock The smallest grammar problem.
\newblock {\em IEEE Transactions on Information Theory}, 51(7):2554--2576,
  2005.

\bibitem{GawryLZ}
Pawe\l{} Gawrychowski.
\newblock Pattern matching in {Lempel-Ziv} compressed strings: fast, simple,
  and deterministic.
\newblock In Camil Demetrescu and Magn{\'u}s~M. Halld{\'o}rsson, editors, {\em
  ESA}, volume 6942 of {\em LNCS}, pages 421--432. Springer, 2011.

\bibitem{RytterSWAT}
Leszek G\k{a}sieniec, Marek Karpi\'nski, Wojciech Plandowski, and Wojciech
  Rytter.
\newblock Efficient algorithms for {Lempel-Ziv} encoding.
\newblock In Rolf~G. Karlsson and Andrzej Lingas, editors, {\em SWAT}, volume
  1097 of {\em LNCS}, pages 392--403. Springer, 1996.

\bibitem{FCPM}
Artur Je\.z.
\newblock Faster fully compressed pattern matching by recompression.
\newblock In Artur Czumaj, Kurt Mehlhorn, Andrew Pitts, and Roger Wattenhofer,
  editors, {\em ICALP (1)}, volume 7391 of {\em LNCS}, pages 533--544.
  Springer, 2012.

\bibitem{fullyNFA}
Artur Je\.z.
\newblock The complexity of compressed membership problems for finite automata.
\newblock {\em Theory of Computing Systems}, pages 1--34, 2013.

\bibitem{onevarlinear}
Artur Je\.z.
\newblock One-variable word equations in linear time.
\newblock In Fedor~V. Fomin, Rusins Freivalds, Marta Kwiatkowska, and David
  Peleg, editors, {\em ICALP (2)}, volume 7966, pages 324--335, 2013.
\newblock full version at http://arxiv.org/abs/1302.3481.

\bibitem{wordequations}
Artur Je\.z.
\newblock {Recompression: a simple and powerful technique for word equations}.
\newblock In Natacha Portier and Thomas Wilke, editors, {\em STACS}, volume~20
  of {\em LIPIcs}, pages 233--244, Dagstuhl, Germany, 2013. Schloss
  Dagstuhl--Leibniz-Zentrum fuer Informatik.

\bibitem{treegrammar}
Artur Je\.z and Markus Lohrey.
\newblock Approximation of smallest linear tree grammar.
\newblock {\em CoRR}, 1309.4958, 2013.
\newblock submitted.

\bibitem{SLPpierwszePM}
Marek Karpi\'nski, Wojciech Rytter, and Ayumi Shinohara.
\newblock Pattern-matching for strings with short descriptions.
\newblock In {\em CPM}, pages 205--214, 1995.

\bibitem{KiefferY96}
John~C. Kieffer and En-Hui Yang.
\newblock Sequential codes, lossless compression of individual sequences, and
  kolmogorov complexity.
\newblock {\em IEEE Transactions on Information Theory}, 42(1):29--39, 1996.

\bibitem{RePair}
N.~Jesper Larsson and Alistair Moffat.
\newblock Offline dictionary-based compression.
\newblock In {\em Data Compression Conference}, pages 296--305. IEEE Computer
  Society, 1999.

\bibitem{Lohreysurvey}
Markus Lohrey.
\newblock Algorithmics on {SLP}-compressed strings: {A} survey.
\newblock {\em Groups Complexity Cryptology}, 4(2):241--299, 2012.

\bibitem{MehlhornSU97}
Kurt Mehlhorn, R.~Sundar, and Christian Uhrig.
\newblock Maintaining dynamic sequences under equality tests in polylogarithmic
  time.
\newblock {\em Algorithmica}, 17(2):183--198, 1997.

\bibitem{Sequitur}
Craig~G. Nevill-Manning and Ian~H. Witten.
\newblock Identifying hierarchical strcture in sequences: A linear-time
  algorithm.
\newblock {\em J. Artif. Intell. Res. (JAIR)}, 7:67--82, 1997.

\bibitem{PlandowskiSLPequivalence}
Wojciech Plandowski.
\newblock Testing equivalence of morphisms on context-free languages.
\newblock In Jan van Leeuwen, editor, {\em ESA}, volume 855 of {\em LNCS},
  pages 460--470. Springer, 1994.

\bibitem{SLPaprox}
Wojciech Rytter.
\newblock Application of {L}empel-{Z}iv factorization to the approximation of
  grammar-based compression.
\newblock {\em Theor. Comput. Sci.}, 302(1-3):211--222, 2003.

\bibitem{SLPaproxSakamoto}
Hiroshi Sakamoto.
\newblock A fully linear-time approximation algorithm for grammar-based
  compression.
\newblock {\em J. Discrete Algorithms}, 3(2-4):416--430, 2005.

\bibitem{SLPapproxNPhard}
James~A. Storer and Thomas~G. Szymanski.
\newblock The macro model for data compression.
\newblock In Richard~J. Lipton, Walter~A. Burkhard, Walter~J. Savitch, Emily~P.
  Friedman, and Alfred~V. Aho, editors, {\em STOC}, pages 30--39. ACM, 1978.

\bibitem{Yaopowers}
Andrew Chi-Chih Yao.
\newblock On the evaluation of powers.
\newblock {\em SIAM J. Comput.}, 5(1):100--103, 1976.

\end{thebibliography}
\end{document}